\documentclass[acmsmall]{acmart}
\makeatletter
\def\@ACM@copyright@check@cc{}
\makeatother

\AtBeginDocument{%
  }

\setcopyright{cc}
\setcctype{by}
\copyrightyear{2024}
\acmYear{2024}
\acmJournal{PACMPL}
\acmVolume{8}
\acmNumber{PLDI}
\acmArticle{184}
\acmMonth{6}
\acmDOI{10.1145/3656414}
\acmSubmissionID{pldi24main-p195-p}
\received{2023-11-16}
\received[accepted]{2024-03-31}

\usepackage{booktabs}   
\usepackage{subcaption} 

\usepackage{hhline}
\usepackage[normalem]{ulem}
\usepackage{mathtools}
\usepackage{blkarray, bigstrut}
\usepackage{graphicx,wrapfig,lipsum}
\usepackage{tcolorbox}
\usepackage{enumitem}
\usepackage{array}
\usepackage{algorithm}
\usepackage{algorithmic}
\usepackage{mathpartir}
\usepackage{tikz}
\usetikzlibrary{shapes,arrows,snakes,decorations.markings}
\usetikzlibrary{svg.path,arrows.meta, calc}

\usetikzlibrary{shapes.geometric}
\usetikzlibrary{arrows.meta,arrows}
\usetikzlibrary{decorations.text}
\usepackage{multirow}
\usepackage[position=b]{subcaption}

\newtheorem{lem}{Lemma}[section]
\newtheorem{thm}{Theorem}[section]
\newtheorem{defn}{Definition}

\theoremstyle{definition}

\numberwithin{subcase}{case}

\numberwithin{subsubcase}{subcase}

\numberwithin{subsubsubcase}{subsubcase}
\newcommand{\st}{~.~}

\definecolor{periwinkle}{rgb}{0.8, 0.8, 1.0}
\definecolor{powderblue}{rgb}{0.69, 0.88, 0.9}
\definecolor{sandstorm}{rgb}{0.93, 0.84, 0.25}
\definecolor{trueblue}{rgb}{0.0, 0.45, 0.81}

\newlength\Origarrayrulewidth

\newcommand{\omitthis}[1]{}

\let\oldstar\star
\renewcommand{\star}{\oldstar}

\newcommand{\im}[1]{\ensuremath{#1}}

\newcommand{\kw}[1]{\im{\mathtt{#1}}}

\tikzstyle{decision} = [diamond, draw, fill=blue!20,
    text width=4.5em, text badly centered, node distance=3cm, inner sep=0pt]
\tikzstyle{block} = [draw, very thick, fill=white, rectangle,
    minimum height=2.5em, minimum width=6em, text centered]

\tikzstyle{line} = [draw, -latex']
\tikzstyle{cloud} = [draw, ellipse,fill=red!20, node distance=3cm,    minimum height=2em]

\tikzstyle{vecArrow} = [thick, decoration={markings,mark=at position
  1 with {\arrow[semithick]{open triangle 60}}},
  double distance=1.4pt, shorten >= 5.5pt,
  preaction = {decorate},
  postaction = {draw,line width=1.4pt, white,shorten >= 4.5pt}]
\tikzstyle{innerWhite} = [semithick, white,line width=1.4pt, shorten >= 4.5pt]

\newcommand{\dist}{P}
\newcommand{\mech}{M}
\newcommand{\univ}{\mathcal{X}}

\newcommand{\qrounds}{r}
\newcommand{\answer}{a}
\newcommand{\sample}{X}
\newcommand{\ex}[2]{{\ifx&#1& \mathbb{E} \else \underset{#1}{\mathbb{E}} \fi \left[#2\right]}}
\newcommand{\pr}[2]{{\ifx&#1& \mathbb{P} \else \underset{#1}{\mathbb{P}} \fi \left[#2\right]}}
\newcommand{\var}[2]{{\ifx&#1& \mathrm{Var} \else \underset{#1}{\mathrm{Var}} \fi \left[#2\right]}}

\newcommand{\from}{:}
\newcommand{\sep}{ \ | \ }

\newcommand{\lin}{\kw{in}}
\newcommand{\lex}{\kw{ex}}
\newcommand{\expr}{e}
\newcommand{\aexpr}{a}
\newcommand{\bexpr}{b}

\newcommand{\qexpr}{\psi}
\newcommand{\qval}{\alpha}
\newcommand{\query}{{\tt query}}
\newcommand{\qquery}{{\tt q}}
\newcommand{\eif}{\;\kw{if}\;}

\newcommand{\ewhile}{\;\kw{while}\;}

\newcommand{\eskip}{\kw{skip}}
\newcommand{\edo}{\;\kw{do}\;}
\newcommand{\esign}{~\kw{sign}~}
\newcommand{\elog}{~\kw{log}~}

\newcommand{\cdom}{\mathcal{C}}
\newcommand{\ldom}{\mathcal{L}}

\newcommand{\config}[1]{\langle #1 \rangle}

\newcommand{\clabel}[1]{\left[ #1 \right]}

\newcommand{\etrue}{\kw{true}}
\newcommand{\efalse}{\kw{false}}

\newcommand{\env}{\rho}

\newcommand{\aarrow}{\Downarrow_a}
\newcommand{\barrow}{\Downarrow_b}
\newcommand{\earrow}{\Downarrow_e}
\newcommand{\qarrow}{\Downarrow_q}

\newcommand{\assign}[2]{ \mathrel{ #1  \leftarrow #2 } }

\newcommand{\vtrace}{\kw{\tau}}

\newcommand{\tdom}{\mathcal{T}}
\newcommand{\trace}{\kw{\tau}}

\newcommand{\vcounter}{\kw{cnt}}

\newcommand{\event}{\kw{\epsilon}}
\newcommand{\eventset}{\mathcal{E}}

\newcommand{\eventdep}{\mathsf{DEP_{\kw{e}}}}
\newcommand{\asn}{\kw{{asn}}}
\newcommand{\test}{\kw{{test}}}

\newcommand{\diff}{\kw{Diff}}

\newcommand{\tracecat}{{\scriptscriptstyle ++}}
\newcommand{\traceadd}{{\small ::}}

\newcommand{\walks}{\mathcal{WK}}

\newcommand{\len}{\kw{len}}
\newcommand{\lvar}{\kw{LV}}

\newcommand{\qvar}{\kw{QV}}

\newcommand{\vardep}{\mathsf{DEP_{var}}}

\newcommand{\tlabel}{\mathbb{TL}}

\newcommand{\traceG}{\kw{{G_{trace}}}}
\newcommand{\traceV}{\kw{{V_{trace}}}}
\newcommand{\traceE}{\kw{{E_{trace}}}}
\newcommand{\traceF}{\kw{{Q_{trace}}}}
\newcommand{\traceW}{\kw{{W_{trace}}}}

\newcommand{\flowsto}{\kw{flowsTo}}
\newcommand{\live}{\kw{RD}}

\newcommand{\progG}{\kw{{G_{est}}}}
\newcommand{\progV}{\kw{{V_{est}}}}
\newcommand{\progE}{\kw{{E_{est}}}}
\newcommand{\progF}{\kw{{Q_{est}}}}
\newcommand{\progW}{\kw{{W_{est}}}}
\newcommand{\progA}{{\kw{A_{est}}}}

\newcommand{\vertxs}{\kw{V}}

\newcommand{\qflag}{\kw{Q}}
\newcommand{\edges}{\kw{E}}
\newcommand{\weights}{\kw{W}}
\newcommand{\qlen}{\len^{\tt q}}

\newcommand{\pathsearch}{\mathsf{AdaptBD}}

\newcommand{\abst}[1]{\kw{abs}{#1}}

\newcommand{\absG}{\abst{\kw{G}}}
\newcommand{\absV}{\abst{\kw{V}}}
\newcommand{\absE}{\abst{\kw{E}}}

\newcommand{\absclr}{{\kw{TB}}}

\newcommand{\constdom}{\mathcal{SC}}
\newcommand{\dcdom}{\mathcal{DC}}

\newcommand{\highlight}[1]{#1}
\newcommand{\review}[1]{}

\newcommand{\todo}[1]{}
\newcommand{\todomath}[1]{}
\newcommand{\jl}[1]{}
\newcommand{\jlside}[1]{}
\newcommand{\dg}[1]{}
\newcommand{\dgside}[1]{}
\newcommand{\mg}[1]{}
\newcommand{\mgside}[1]{}
\newcommand{\wq}[1]{}
\newcommand{\wqside}[1]{}

\newcommand{\THESYSTEM}{\textsf{AdaptFun}}

\begin{document}

\title{Program Analysis for Adaptive Data Analysis}

\author{Jiawen Liu}
\orcid{0009-0009-0900-9910}
\affiliation{%
  \institution{Boston University}
  \city{Boston}
  \country{USA}}
\email{jiawenl@bu.edu}

\author{Weihao Qu}
\orcid{0000-0003-1027-6556}
\affiliation{%
  \institution{Monmouth University}
  \city{West Long Branch}
  \country{USA}}
\email{wqu@monmouth.edu}

\author{Marco Gaboardi}
\orcid{0000-0002-5235-7066}
\affiliation{%
  \institution{Boston University}
  \city{Boston}
  \country{USA}}
\email{gaboardi@bu.edu}

\author{Deepak Garg}
\orcid{0000-0002-0888-3093}
\affiliation{%
  \institution{MPI-SWS}
  \city{MPI-SWS}
  \country{Germany}}
\email{dg@mpi-sws.org}

\author{Jonathan Ullman}
\orcid{0000-0002-0323-5564}
\affiliation{%
  \institution{Northeastern University}
  \city{Boston}
  \country{USA}}
\email{jullman@gmail.com}

\begin{abstract}
  Data analyses are usually designed to identify some property of the population from which the data are drawn, generalizing beyond the specific data sample. For this reason, data analyses are often designed in a way that guarantees that they produce a low generalization error.
  That is, they are designed so that the result of a data analysis run on a sample data does not differ too much from the result one would achieve by running the analysis over the entire population.

  An adaptive data analysis can be seen as a process composed by multiple queries interrogating some data, where the choice of which query to run next may rely on the results of previous queries.
  The generalization error of each individual query/analysis can be controlled by using an array of well-established statistical techniques.
  However, when queries are arbitrarily composed, the different errors can propagate through the chain of different queries and bring to a high generalization error.
  To address this issue, data analysts are designing several techniques that not only guarantee bounds on the generalization errors of single queries, but that also guarantee bounds on the generalization error of the composed analyses.
  The choice of which of these techniques to use, often depends on the chain of queries that an adaptive data analysis can generate.

  In this work, we consider adaptive data analyses implemented as while-like programs and we design a program analysis which can help with identifying which technique to use to control their generalization errors.
  More specifically, we formalize the intuitive notion of \emph{adaptivity} as a quantitative property of programs.
  We do this because the adaptivity level of a data analysis is a key measure to choose the right technique.
  Based on this definition, we design a program analysis for soundly approximating this quantity.
  The program analysis generates a representation of the data analysis as a weighted dependency graph, where the weight is an upper bound on the number of times each variable can be reached, and uses a path search strategy to guarantee an upper bound on the adaptivity.
  We implement our program analysis  and show that it can help to analyze the adaptivity of several concrete data analyses with different adaptivity structures.
  \end{abstract}



\maketitle

\section{Introduction}
\label{sec:intro}
Consider a dataset $X$ consisting of $n$ independent samples from some unknown population $\dist$.  How can we ensure that the conclusions drawn from $X$ \emph{generalize} to the population $\dist$?  Despite decades of research in statistics and machine learning on methods for ensuring generalization, there is an increased recognition that many scientific findings generalize poorly (e.g.
\cite{Ioannidis05,GelmanL13}
).  While there are many reasons a conclusion might fail to generalize, one that is receiving increasing attention is \emph{adaptivity}, which occurs when the choice of method for analyzing the dataset depends on previous interactions with the same dataset~\cite{GelmanL13}.
 Adaptivity can arise from many common practices, such as exploratory data analysis, using the same data set for feature selection and regression, and the re-use of datasets across research projects.  Unfortunately, adaptivity invalidates traditional methods for ensuring generalization and statistical validity, which assume that the method is selected independently of the data. The misinterpretation of adaptively selected results has even been blamed for a ``statistical crisis'' in empirical science~\cite{GelmanL13}.

\begin{figure}
    \centering
    \includegraphics[width=0.7\columnwidth]{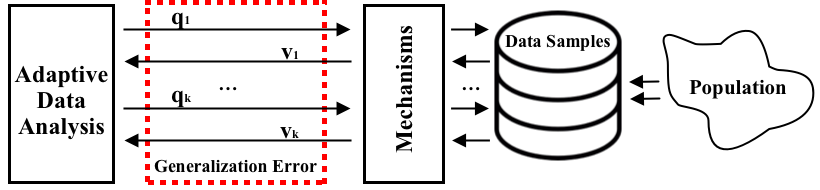}
    \caption{Overview of our Adaptive Data Analysis model. We have a population that we are interested in studying, and a dataset containing individual samples from this population. The adaptive data analysis we are interested in running has access to the dataset through queries of some pre-determined family (e.g., statistical or linear queries) mediated by a mechanism. This mechanism uses randomization to reduce the generalization error of the queries issued to the data.}
    \label{fig:adaptivity-model-overview}
\vspace{-0.5cm}
\end{figure}

A line of work initiated by \citet{DworkFHPRR15}, \citet{HardtU14} posed the question: Can we design \emph{general-purpose} methods that ensure generalization in the presence of adaptivity, together with guarantees on their accuracy?  The idea that has emerged in these works is to use randomization to help ensure generalization. Specifically, these works have proposed to mediate the access of an adaptive data analysis to the data by means of queries from some pre-determined family (we will consider here a specific family of queries often called "statistical" or "linear" queries) that are sent to a  \emph{mechanism} which uses some randomized process to guarantee that the result of the query does not depend too much on the specific
sampled dataset. This guarantees that the result of the queries generalizes well. This approach is described in Fig.~\ref{fig:adaptivity-model-overview}.
This line of work has identified many new algorithmic techniques for ensuring generalization in adaptive data analysis, leading to algorithms with greater statistical power than all previous approaches. Common methods proposed by these works include, the addition of noise to the result of a query, data splitting, using sampling methods, etc. Moreover, these works have also identified problematic strategies for adaptive analysis, showing limitations on the statistical power one can hope to achieve. Subsequent works have then further extended the methods and techniques in this approach and further extended the theoretical underpinning of this approach, e.g.~\cite{dwork2015reusable,dwork2015generalization,BassilyNSSSU16,UllmanSNSS18,FeldmanS17,jung2019new,SteinkeZ20,RogersRSSTW20,DaganK22,Blanc23}.

A key development in this line of work is that the best method for ensuring generalization in an adaptive data analysis depends to a large extent on the number of \emph{rounds of adaptivity}, the depth of the chain of queries. As an informal example, the program $x \leftarrow q_1(D);y \leftarrow q_2(D,x);z \leftarrow q_3(D,y)$ has three rounds of adaptivity, since $q_2$  depends on $D$ not only directly because it is one of its input but also via the result of $q_1$, which is also run on $D$, and similarly,  $q_3$ depends on $D$ directly but also via the result of $q_2$, which in turn depends on the result of $q_1$. The works we discussed above showed that, not only does the analysis of the generalization error depend on the number of rounds, but knowing the number of rounds actually allows one to choose methods that lead to the smallest possible generalization error.
For instance, these works showed that when an adaptive data analysis uses a large number of rounds of adaptivity, then a low generalization error can be achieved by the Gaussian mechanism
adding to the result of each query Gaussian noise scaled to the number of rounds. When instead  an adaptive data analysis uses a small number of rounds of adaptivity then a low generalization error can be achieved by using more specialized methods, such as the data splitting mechanism or the reusable holdout technique from~\citet{DworkFHPRR15}.

To better understand this idea, we show in Fig.~\ref{fig:generalization_errors} three experiments showcasing these situations:
In Fig.~\ref{fig:generalization_errors}(a) we show the results of a specific analysis\footnote{We will use formally a program implementing this analysis (Fig.~\ref{fig:overview-example}) as a running example in the rest of the paper.}
with two rounds of adaptivity.

 \highlight{This analysis can be seen as a classifier over a population of 400 attributes and 1 label. Using a dataset sampled from this population, this analysis first runs 400 non-adaptive queries on the first 400 attributes of this dataset,
  computing correlations between each attribute and the label. Then, it runs the last query depending on all these correlations. The adaptivity of this anlaysis is $2$ because only the last query relies on the results of its previous queries' results.
  Without any mechanism the generalization error of the last query is pretty large, and the lower generalization error is achieved when the data-splitting method is used.
  In Fig.~\ref{fig:generalization_errors}(c), we use the same data analysis program as in Fig.~\ref{fig:generalization_errors}(a). It runs over a dataset with larger data size and shows only the root mean square error of the last \emph{adaptive} query when the
total query number varies along the x-axis. The result mostly accounts for
the fact that total query numbers affect the generalization error as well.
Fig.~\ref{fig:generalization_errors}(c) also highlights that different mechanisms, for the same analysis, produce results with different generalization errors, and
using a larger data set changes the magnitude of the error. }
In Fig.~\ref{fig:generalization_errors}(b), we show the results of a specific analysis\footnote{We will present this analysis formally in Section~\ref{sec:examples}.} with four hundred rounds of adaptivity.
At each step, this analysis runs an adaptive query based on the results of the previous ones. Without any mechanism, the generalization error of most of the queries is pretty large, and this error can be lowered by using Gaussian noise.
\highlight{We use in total three different mechanisms here: Gaussian mechanism which adds the result of each query Gaussian noise; Data Splitting mechanism
that splits the data into a few parts; Thresholdout mechanism that uses the reusable holdout technique from~\citet{DworkFHPRR15}.  }
{\small
\begin{figure}
\centering
\begin{subfigure}{.322\textwidth}
\begin{centering}
\includegraphics[width=1.0\textwidth]{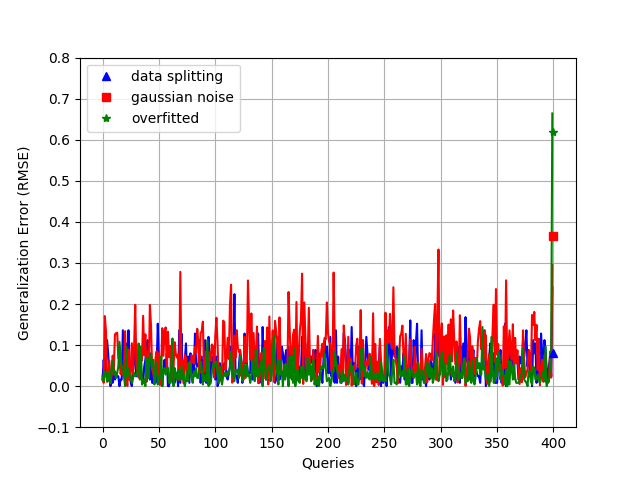}
\caption{}
\end{centering}
\end{subfigure}
\quad
\begin{subfigure}{.322\textwidth}
\begin{centering}
\includegraphics[width=1.0\textwidth]{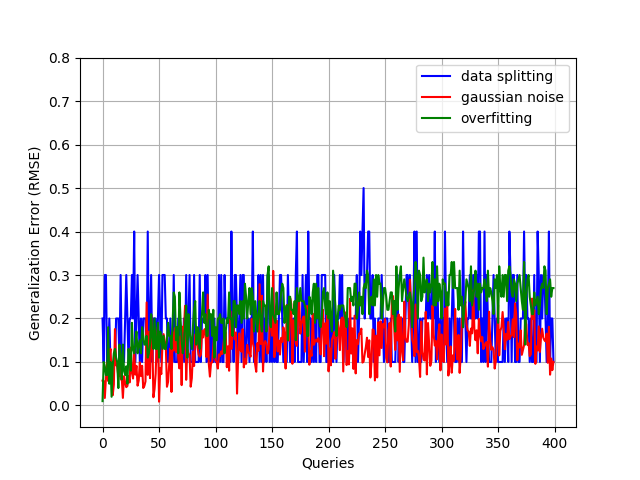}
\caption{}
\end{centering}
\end{subfigure}
\begin{subfigure}{.322\textwidth}
\begin{centering}
\includegraphics[width=1.0\textwidth]{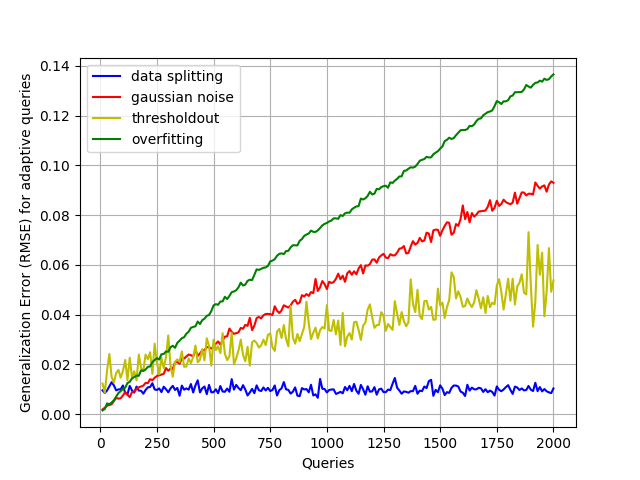}
\caption{}
\end{centering}
\end{subfigure}
\vspace{-0.2cm}
 \caption{
 The generalization errors of two adaptive data analysis examples, under different choices of mechanisms.
 (a)  2 rounds adaptivity.
 (b)  400 rounds adaptivity.
 (c) 2 rounds adaptivity with varied query numbers.
}
\label{fig:generalization_errors}
\vspace{-0.2cm}
\end{figure}
}

This scenario motivates us to explore the design of program analysis techniques that can be used to estimate the number of \emph{rounds of adaptivity} that a program implementing a data analysis can perform. These techniques could be used to help a data analyst in the choice of the mechanism to use,
and they
could ultimately be integrated into a tool for adaptive data analysis such as the \emph{Guess and Check} framework by~\citet{RogersRSSTW20}.

The first problem we face is \emph{how to formally define} a model for adaptive data analysis which is general enough to support the methods we discussed above and which would permit to formulate the notion of adaptivity these methods use. We take the approach of designing a programming framework for submitting queries to some \emph{mechanism} giving access to the data mediated by one of the techniques we mentioned before, e.g., adding Gaussian noise, randomly selecting a subset of the data, using the reusable holdout technique, etc. In this approach, a program models an \emph{analyst} asking a sequence of queries to the mechanism. The mechanism runs the queries on the data applying one of the methods above and returns the result to the program. The program can then use this result to decide which query to run next. Overall, we are interested in controlling the generalization error of the query results returned by the mechanism, by means of the adaptivity.

The second problem we face is \emph{how to define the adaptivity of a given program}.
Intuitively, a query $Q$ may depend on another query $P$, if there are two values that $P$ can return which affect in different ways the execution of $Q$.
For example, as shown in \cite{dwork2015reusable}, and as we did in our example in Fig.~\ref{fig:generalization_errors}(a), one can design a machine learning algorithm for constructing a classifier which first computes each feature's correlation with the label via a sequence of queries, and then constructs the classifier based on the correlation values. If one feature's correlation changes, the classifier depending on features is also affected.
This notion of dependency builds on the execution trace as a \emph{causal history}. In particular, we are interested in the history or provenance of a query up until this is executed, and we are not then concerned about how the result is used --- except for tracking whether the result of the query may further cause some other queries. This is because we focus on the generalization error of queries and not their post-processing. %
To formalize this intuition as a quantitative program property,
we use a trace semantics recording the execution history of programs on some given input --- and we create a dependency graph, where the dependency between different variables (queries are also assigned to variables) is explicit and tracks which variable is associated with a query request. We then enrich this graph with weights describing the number of times each variable is evaluated in a program evaluation starting with an initial state. The adaptivity is then defined as the length of the walk visiting the most query-related variables on this graph\footnote{Formally, graphs will be well-defined only for terminating programs, this will guarantee that the walk is finite}. In other words, we define adaptivity as a \emph{quantitative form of program dependency}.

The third problem we face is \emph{how to estimate the adaptivity of a given program}.
The adaptive data analysis model we consider and our definition of adaptivity suggest that for this task we can use a  program analysis that is based on some form of dependency analysis. This analysis needs to take into consideration:
1) the fact that, in general, a query $Q$ is not a monolithic block but rather it may depend, through the use of variables and values, on other parts of the program. Hence, it needs to consider some form of data flow analysis.
2) the fact that, in general, the decision on whether to run a query or not may depend on some other value. Hence,
 it needs to consider some form of control flow analysis.
 3) the fact that, in general, we are not only interested in whether there is a dependency or not, but in the length of the chain of dependencies. Hence, it needs to consider some quantitative information about the program dependencies.

To address these considerations and be able to estimate a sound upper bound on the adaptivity of a program,
we develop a static program analysis algorithm, named {\THESYSTEM}, which combines data flow and control flow analysis with reachability bound analysis~\cite{GulwaniZ10}. This combination gives tighter bounds on the adaptivity of a program than the ones one would achieve by directly using the data and control flow analyses or the ones that one would achieve by directly using reachability bound analysis techniques alone. We evaluate {\THESYSTEM} on a number of examples showing that it is able to efficiently estimate precise upper bounds on the adaptivity of different programs.
All the proofs and extended definitions can be found in the supplementary material.

To summarize, our work aims at the design of a static analysis for programs implementing adaptive analysis that can estimate their rounds of adaptivity. Specifically, our contributions are:
\begin{enumerate}
    \item A programming framework for adaptive data analyses where programs represent analysts that can query generalization-preserving mechanisms mediating the access to some data.
    \item
    A formal definition of the notion of adaptivity under the analyst-mechanism model,
    built on a variable-based dependency graph constructed using sets of program execution traces.
    \item
    A static program analysis algorithm {\THESYSTEM} combining data flow, control flow and  reachability bound analysis in order to provide tight bounds on the adaptivity of a program.
    \item A soundness proof of the program analysis showing that the adaptivity estimated by {\THESYSTEM} bounds the true adaptivity of the program.
    \item A prototype implementation of {\THESYSTEM} and an experimental evaluation showing its accuracy and efficiency of the adaptivity estimation on several examples.
    We also provide an evaluation showing how the generalization error of several real-world data analyses can be effectively reduced by using the information provided by {\THESYSTEM}.
\end{enumerate}

\section{Overview}
\label{sec:overview}
\subsection{Some results in Adaptive Data Analysis}
In Adaptive Data Analysis, an \emph{analyst} is interested in studying some distribution $\dist$ over some domain $\univ$.  Following previous works~\cite{DworkFHPRR15,HardtU14,BassilyNSSSU16}, we focus on the setting where the analyst is interested in answers to \emph{statistical queries} (also known as \emph{linear queries}) over the distribution.  A statistical query is usually defined by some function $\qquery \from \univ \to [-1,1]$ (often other codomains such as $[0,1]$ or $[-R,+R]$, for some $R$, are considered).  The analyst wants to learn the \emph{population mean}, which is defined as
$\qquery(\dist) = \ex{\sample \sim \dist}{\qquery(\sample)}$.
We assume that the distribution $\dist$ can only be accessed via a set of \emph{samples} $\sample_1,\dots,\sample_n$ drawn independently and identically distributed (i.i.d.) from $\dist$.  These samples are held by a mechanism $\mech(\sample_1,\dots,\sample_n)$ who receives the query $\qquery$ and computes an answer
$\answer \approx \qquery(\dist)$.
The na\"ive way to approximate the population mean is to use the \emph{empirical mean}, which (abusing notation) is defined as
$\qquery(\sample_1,\dots,\sample_n) = \frac{1}{n} \sum_{i=1}^{n} \qquery(X_i)$.
However, the mechanism $M$ can adopt some methods for improving the generalization error $| a- \qquery(\dist)|$.

In this work we consider analysts that ask a sequence of $k$ queries $\qquery_1,\dots,\qquery_k$.  If the queries are all chosen in advance, independently of the answers $a_1,\dots,a_k$ of each other, then we say they are \emph{non-adaptive}.  If the choice of each query $\qquery_j$ depends on the prefix $\qquery_1,\answer_1,\dots,\qquery_{j-1},\answer_{j-1}$ then they are \emph{fully adaptive}.  An important intermediate notion is \emph{$\qrounds$-round adaptive}, where the sequence can be partitioned into $\qrounds$ batches of non-adaptive queries.  Note that non-adaptive queries are $1$-round and fully adaptive queries are $k$-round adaptive.

We now review what is known about the problem of answering $r$-round adaptive queries.
\begin{thm}[\cite{BassilyNSSSU16}]
\label{thm:nonadapt-adapt}
\begin{enumerate}

\item For any distribution $\dist$, and any $k$ \emph{non-adaptive} statistical queries, with high probability,
$
\max_{j=1,\dots,k} | \answer_j - \qquery_j(\dist) | = O\left( \sqrt{\frac{\log k}{n}}  \right)
$.
\item For any distribution $\dist$, and  any $k$  \emph{$\qrounds$-round adaptive} statistical queries, with $\qrounds \geq 2$, with high probability, the empirical mean (rounded to an appropriate number of bits of precision)\footnote{With infinite precision even two queries may give unbounded error, when the first query's result encodes the whole data.} satisfies:\\
$
\max_{j=1,\dots,k} | \answer_j - \qquery_j(\dist) | = O\left( \sqrt{  \frac{k}{n}}  \right)
$
\end{enumerate}
\end{thm}
In fact, these bounds are tight (up to constant factors) which means that even allowing one extra round of adaptivity leads to an exponential increase in the generalization error, from $\log k$ to $k$.

\citet{DworkFHPRR15} and \citet{BassilyNSSSU16} showed that by using carefully calibrated Gaussian noise in order to limit the dependency of a single query on the specific data instance, one
can actually achieve much stronger generalization error as a function of the number of queries, specifically.
\begin{thm}[\cite{DworkFHPRR15, BassilyNSSSU16}] \label{thm:gaussiannoise} For any distribution $\dist$, any $k$, any $\qrounds \geq 2$ and any \emph{$\qrounds$-round adaptive} statistical queries, if we answer queries with carefully calibrated Gaussian noise, with high probability,  we have:
\begin{center}
  $
\max_{j=1,\dots,k} | \answer_j - \qquery_j(\dist) | = O\left( \frac{\sqrt[4]{k}}{\sqrt{n}}  \right)
$
\end{center}
\end{thm}
More interestingly, \citet{DworkFHPRR15}
also gave a refined bounds that can be achieved with different mechanisms depending on the number of rounds of adaptivity.   \begin{thm}[\cite{DworkFHPRR15}] \label{thm:gaussiannoise2} For any $r$ and $k$, there exists a mechanism such that for any distribution $\dist$, and any $\qrounds \geq 2$ any \emph{$\qrounds$-round adaptive} statistical queries, with high probability, it satisfies
\begin{center}
  $
\max_{j=1,\dots,k} | \answer_j - \qquery_j(\dist) | = O\left( \frac{r \sqrt{\log k}}{\sqrt{n}}  \right)
$
\end{center}
\end{thm}
Notice that Theorem~\ref{thm:gaussiannoise2} has different quantification in that the optimal choice of mechanism depends on the number of queries {and number of rounds of adaptivity}.  This suggests that if one knows a good \emph{a priori upper bound on the number of rounds of adaptivity}, one can choose the appropriate mechanism and get a much better guarantee in terms of the generalization error.
As an example, as we can see in Fig.~\ref{fig:generalization_errors}, if we know that an algorithm is 2-rounds adaptive, we can choose data splitting as {the} mechanism, while if we know that an algorithm has many rounds of adaptivity we can choose Gaussian noise. It is worth to stressing that by knowing the number of rounds of adaptivity one can also compute a concrete upper bound on the generalization error of a data analysis. This information allows one to have a quantitative, a priori, estimation of the effectiveness of a data analysis.
This motivates us to design a static program analysis aimed at giving good \emph{a priori} upper bounds on the number of rounds of adaptivity of a program.

{\small
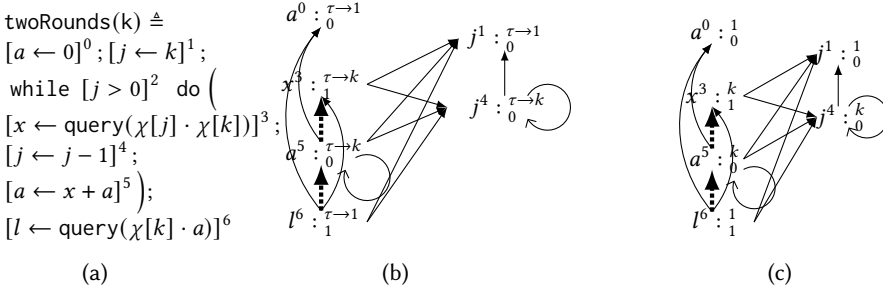
\begin{figure}
\centering
\begin{subfigure}{.2\textwidth}
\begin{centering}
$
    \begin{array}{l}
    \kw{twoRounds(k)} \triangleq \\
           \clabel{ \assign{a}{0}}^{0} ;
            \clabel{\assign{j}{k} }^{1} ; \\
            \ewhile ~ \clabel{j > 0}^{2} ~ \edo ~ \Big( \\
             \clabel{\assign{x}{\query(\chi[j] \cdot \chi[k])} }^{3}  ; \\
             \clabel{\assign{j}{j-1}}^{4} ;\\
            \clabel{\assign{a}{x + a}}^{5}       \Big);\\
            \clabel{\assign{l}{\query(\chi[k]\cdot a)} }^{6}\\
        \end{array}
$
\caption{}
\end{centering}
\end{subfigure}
\begin{subfigure}{.36\textwidth}
\qquad
\begin{centering}
\begin{tikzpicture}[scale=\textwidth/16.5cm,samples=200]
\draw[] (0, 10) circle (0pt) node
{{ $a^0: {}^{ \tau \to 1}_{0}$}};
\draw[] (0, 7) circle (0pt) node
{\textbf{$x^3: {}^{\tau \to k}_{1}$}};
\draw[] (0, 4) circle (0pt) node {{ $a^5: {}^{\tau \to k}_{0}$}};
\draw[] (0, 1) circle (0pt) node
{{ $l^6: {}^{\tau \to  1}_{1}$}};
\draw[] (8, 9) circle (0pt) node {\textbf{$j^1: {}^{\tau \to 1}_{0}$}};
\draw[] (8, 6) circle (0pt) node {{ $j^4: {}^{ \tau \to k}_{0}$}};
\draw[ ultra thick, -latex, densely dotted,] (0, 1.5)  -- (0, 3.5) ;
\draw[ ultra thick, -latex, densely dotted,] (0, 4.5)  -- (0, 6.5) ;
\draw[ -latex] (0, 4.5)  to  [out=-230,in=230]  (0, 9.5) ;
\draw[ -Straight Barb] (1.5, 3.8) arc (120:-200:1);
\draw[ -Straight Barb] (9, 6.5) arc (150:-150:1);
\draw[ -latex] (8, 6.5)  -- (8, 8.5) ;
\draw[ -latex] (0, 1.5)  to  [out=-230,in=230]  (0, 9.5) ;
\draw[ -latex] (2, 7)  -- (6, 9) ;
\draw[ -latex] (2, 7)  -- (5.5, 6) ;
\draw[ -latex] (2, 4.5)  -- (6, 9) ;
\draw[ -latex] (2, 4.5)  -- (5.5, 6) ;
\draw[ -latex] (2, 1)  -- (6, 9) ;
\draw[ -latex] (2, 1)  -- (5.5, 6) ;
\draw[ -latex] (0, 1.5)  to  [out=50,in=-50]  (0, 6.5) ;
\end{tikzpicture}
\caption{}
\end{centering}
\end{subfigure}
   \begin{subfigure}{.36\textwidth}
   \begin{centering}
   \begin{tikzpicture}[scale=\textwidth/18cm,samples=200]
\draw[] (0, 10) circle (0pt) node
{{ $a^0: {}^1_{0}$}};
\draw[] (0, 7) circle (0pt) node
{\textbf{$x^3: {}^{k}_{1}$}};
\draw[] (0, 4) circle (0pt) node
{{ $a^5: {}^{k}_{0}$}};
\draw[] (0, 1) circle (0pt) node
{{ $l^6: {}^{1}_{1}$}};
\draw[] (6, 9) circle (0pt) node {\textbf{$j^1: {}^{1}_{0}$}};
\draw[] (6, 6) circle (0pt) node {{ $j^4: {}^{k}_{0}$}};
\draw[ ultra thick, -latex, densely dotted,] (0, 1.5)  -- (0, 3.5) ;
\draw[ ultra thick, -latex, densely dotted,] (0, 4.5)  --
(0, 6.5) ;
\draw[  -latex] (0, 4.5)  to  [out=-230,in=230]
(0, 9.5) ;
\draw[  -Straight Barb] (1.5, 3.5) arc (120:-200:1);
\draw[  -Straight Barb] (6.5, 6.5) arc (150:-150:1);
\draw[  -latex] (6, 6.5)  -- (6, 8.5) ;
\draw[ -latex] (1.5, 7)  -- (5, 9) ;
\draw[ -latex] (1.5, 4)  -- (5, 9) ;
\draw[ -latex] (1.5, 7)  -- (5, 6) ;
\draw[ -latex] (1.5, 4)  -- (5, 6) ;
\draw[  -latex] (0, 1.5)  to  [out=-230,in=230]  (0, 9.5) ;
\draw[ -latex] (2, 1)  -- (5, 9) ;
\draw[ -latex] (2, 1)  -- (5, 6) ;
\draw[ -latex] (0, 1.5)  to  [out=50,in=-50]  (0, 6.5) ;
\end{tikzpicture}
\caption{}
   \end{centering}
   \end{subfigure}
\vspace{-0.4cm}
 \caption{(a) The program $\kw{twoRounds(k)}$, an example
with two rounds of adaptivity (b) The corresponding semantics-based dependency graph (c) The estimated dependency graph from $\THESYSTEM$.
}
\label{fig:overview-example}
\vspace{-0.5cm}
\end{figure}
}

\subsection{ {\THESYSTEM} formally through an example.}
We illustrate the key technical components of our framework through a simple adaptive data analysis with two rounds of adaptivity.
In this analysis, an analyst asks $k+1$ queries to a mechanism in two phases.
In the first phase, the analyst asks $k$ queries and stores the answers that are provided by the mechanism. In the second phase, the analyst constructs a new query based on the results of the previous $k$ queries and sends this query to the mechanism.
The mechanism is abstract here and our goal is to use static analysis to provide an upper bound on adaptivity to help choose the mechanism.
This data analysis assumes that the data domain $\univ$
contains at least $k$ numeric attributes
(every query in the first phase focuses on one), which we index just by natural numbers.
The implementation of this data analysis in the language of {\THESYSTEM} is presented in Fig.~\ref{fig:overview-example}(a).

The {\THESYSTEM} language extends a standard while language\footnote{Programs components are labeled, so that we can uniquely identify every component.} with a query request constructor denoted $\query$.
 Queries have the form $\query(\qexpr)$, where $\qexpr$ is a special expression (see syntax in Section~\ref{sec:loop_language})
representing a function $\from \univ \to U$ on rows of an hidden database that is only accessible through the mechanisms. The domain $\univ$ of this function is the (arbitrary) domain of rows of the database. The codomain $U$ of this function is the query output space which, depending on the specific program, could be $[-1,1]$, $[0,1]$ or $[-R,+R]$, for some $R$. We use this formalization because we are interested in linear queries which, as we discussed in the previous section, compute the empirical mean of functions on rows.
 As an example, $x \leftarrow \query(\chi[j] \cdot \chi[k])$ computes an approximation, according to the used mechanism, of the empirical mean of the product of the $j^{th}$ attribute and $k^{th}$ attribute, identified by $\chi[j] \cdot \chi[k]$. Notice that we don't materialize the mechanism but we assume that it is implicitly run when we execute the query.
 In Fig.~\ref{fig:overview-example}(a), the queries inside the while loop correspond to the first phase of the data analysis and compute the sum of the empirical mean of
the product of the $j$th attribute with the $k$th attribute.
The query outside the loop corresponds to the second phase and computes an approximation of the empirical mean of the last attribute weighted by the sum of the empirical mean of the first $k$ attributes.

This example is intuitively 2-rounds adaptive since we have two clearly distinguished phases, and the queries that we ask in the first phase do not depend on each other (the query $\chi[j] \cdot \chi[k]$ at line $3$ only relies on the counter $j$ and input $k$), while the last query
(at line 6) depends on the results of all the previous queries.
However, capturing this concept formally is surprisingly challenging. The difficulty comes from the quantitative nature of this concept and how this quantitative nature interacts with data and control dependency. We describe how we capture it next.

\subsubsection{Adaptivity definition}
\label{sec:adaptivity-informal}

The central property we are after in this work is the \emph{adaptivity of a program}. We define formally this notion in three steps (details in Section~\ref{sec:adaptivity}). First, we define a notion of dependency, or better \emph{may-dependency}, between variables. To do this we take inspiration from previous works on dependency analysis and information flow control and we say that a variable \emph{may depend} on another one if changing the execution of the latter can affect the execution of the former.
We can see in Fig.~\ref{fig:overview-example}(a) that the value of the variable $l$, which corresponds to the result of the execution of the query in the second phase (in the command with label 6), is affected by the value of the variable $x$, which corresponds to the result of the execution of the query at line 3 in the first phase, via the variable $a$.
To formally define this notion of dependency, as in information flow control, we use the execution history of programs recorded by a trace semantics (see Definition~\ref{def:var_dep}).

Second, we build an annotated weighted directed graph representing the possible dependencies between labeled variables. We call this graph the \emph{semantics-based dependency graph} to stress that this graph summarizes the dependencies we could see if we knew the overall behavior of the program.
The vertices of the graph are the assigned program variables with the label of their assignments, edges are pairs of labeled variables which satisfy the dependency relations, weights are functions associated with vertices and describe the number of times the assignment corresponding to the vertex is executed when the program is run in a given starting state\footnote{In our trace semantics the state is recorded in the trace, so an initial state is actually represented by an initial trace. We will use this terminology in later sections.}, and the annotations, which we call \emph{query annotations}, are bits associated with vertices and describe if the corresponding assignment comes from a query (1) or not (0).
The \emph{semantics-based dependency graph} of the $\kw{twoRounds(k)}$ program
we gave in Fig.~\ref{fig:overview-example}(a) is described in Fig.~\ref{fig:overview-example}(b) (we use dashed arrows for two edges that will be highlighted in the next step, for the moment these can be considered similar to the other edges---i.e. solid arrows).

We have all the variables that are assigned in the program with their labels, and edges representing dependency relations between them.
For example, we have two edges $(l^6, a^5)$ and $(a^5, x^3)$ describing the dependency between the variables assigned by queries. The vertices $l^6$ and $x^3$ are the only ones with query annotation $1$ (the subscript), since they are the only two variables that are in assignments involving  queries. Notice that the graph contains cycles---in this example it contains two self-loops. These cycles capture the fact that the variables $a^5$ and $j^4$ are updated at every iteration of the loop using their previous values. Cycles are essential to capture mutual dependencies like the ones that are generated in loops. Adaptivity is a quantitative notion, so capturing this form of dependencies is not enough.
This is why we also use weights.
\highlight{The weight of a vertex is a function that given an initial state returns a natural number representing
the number of times this vertex is visited during the program execution starting in this initial state, which relies on
the actual value of the input variables defined in the initial trace.
For instance, in the running example in Fig.~\ref{fig:overview-example}(b), the times of the vertex $a^5$ will be visited depends on actual value of
the input variable $k$. To capture this, we define a function
$\kw{lastVal}(\tau, x)$ that will return the latest value of variable $x$ in any initial trace $\trace$. Also,
it is natural that some vertex will be visited only in constant times, such as the vertex $l^{6}$ being executed only once.
 To express this, our $\kw{lastVal}$ function also accepts constant numbers besides input variables, such that we use $\kw{lastVal}(\tau, 1)$ as the weight for  the vertex $l^{6}$.
 In Fig.~\ref{fig:overview-example}(b), we use  $ \tau \to x$ and $\tau \to 1$ as the short notations for
 $\kw{lastVal}(\tau, x)$ and $\kw{lastVal}(\tau, 1)$.
}

Third, we can finally define adaptivity using the semantics-based dependency graph. We actually define this notion with respect to an initial state $\tau$, since different states can give very different adaptivities.
We consider
any  walk  that visits any vertex $v$ of the semantics-based dependency graph no more than the value specified by the vertex's weight $w_v$ and the initial state $\tau$, and that visits a maximal number of query vertices.
The number of query vertices visited is the adaptivity of the program with respect to $\tau$.
In Fig.~\ref{fig:overview-example}(b), assuming that $\tau(k) \geq 1$, we can see that the
walk along the dashed arrows,  $l^{6} \to a^5 \to x^3 $ has two vertices with query annotation $1$, and we cannot find another walk having more than $2$ query vertices, although there is another walk, $l^{6} \to x^3 $, which has $2$ query vertices. So the adaptivity of the program in Fig.~\ref{fig:overview-example}(a) with respect to $\tau$ is $2$. If we consider an initial state $\tau$ such that $\tau(k)=0$ we have that the adaptivity with respect to $\tau$ is instead $1$.

\subsubsection{Static analysis}

To compute statically a sound and accurate upper bound on the \emph{adaptivity} of a program $c$,
we design a program analysis framework named {\THESYSTEM} (formally in Section \ref{sec:algorithm}).
The structure of {\THESYSTEM} (Fig.~\ref{fig:adaptfun}) reflects in part the definition of adaptivity we discussed. Specifically, {\THESYSTEM} is composed by two algorithms (the ones in dashed boxes in the figure), one for building a dependency graph, called \emph{estimated dependency graph}, and the other to estimate the adaptivity from this graph.
The first algorithm generates the \emph{estimated dependency graph} using several program analysis techniques. Specifically,
 {\THESYSTEM} extracts the vertices and the query annotations by looking at the assigned variables of the program, it estimates the edges by using control flow and data flow analysis, and it estimates the weights by using symbolic reachability-bound analysis---weights in this graph are symbolic expressions over input variables.
The second algorithm estimates the
walk which respects the weights and which visits the maximal number of query vertices.
The two algorithms together gives us an  upper bound on the program's \emph{adaptivity}.

 \begin{figure}
  \centering
\includegraphics[width=1.0\columnwidth]{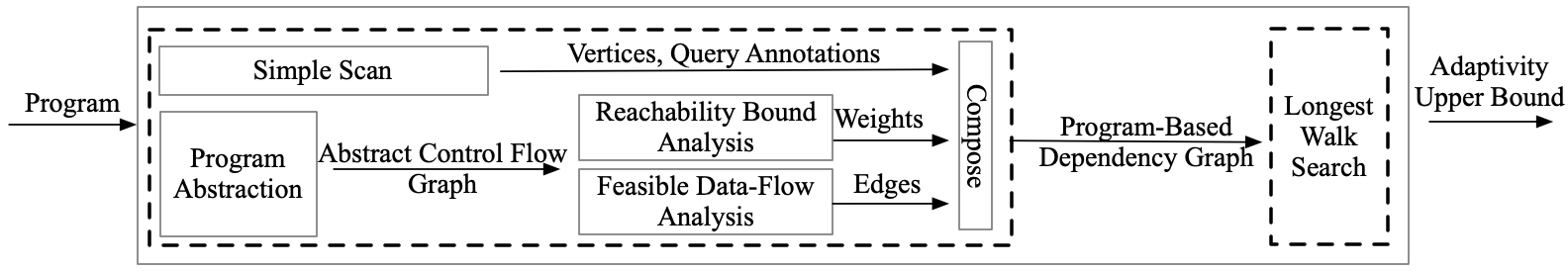}
  \vspace{-0.8cm}
  \caption{The overview of {\THESYSTEM}}
  \label{fig:adaptfun}
  \vspace{-0.5cm}
\end{figure}

We show in Fig.~\ref{fig:overview-example}(c) the estimated dependency graph that our static analysis algorithm returns for the program $\kw{twoRounds(k)}$ in Fig.~\ref{fig:overview-example}(a).
Vertices and query annotations are the same as the ones in Fig.~\ref{fig:overview-example}(b), simply inferred by scanning the program.
 Every edge in Fig.~\ref{fig:overview-example}(b) is precisely inferred by our combined data flow and control flow analysis, this is why Fig.~\ref{fig:overview-example}(c) contains exactly the same edges.
The weight of every vertex is computed using a reachability-bound estimation algorithm which outputs a symbolic expression over the input variables, representing an upper bound on the number of times each assignment is executed.
For example, consider the vertex $x^{3}$, its weight is $k$ and this provides an upper bound on the value returned by the weight function \highlight{$\kw{lastVal}(\trace,k)$} associated with vertex $x^{3}$ in Fig.~\ref{fig:overview-example}(b) for any initial state.

The algorithm searching for the walk first finds a path $l^6:{}^1_1 \to a^5: {}^k_0 \to x^3: {}^k_1$, and then constructs a walk based on this path. Every vertex on this walk is visited once, and the number of vertices with query annotation $1$ in this walk is $2$, which is the upper bound we expect.
{It is worth noting here that $x^3$ and $a^5$ can only be visited once because there isn't an edge to go back to them, even though they both have the weight $k$}.  So the algorithm $\pathsearch$ computes the upper bound $2$ instead of $2k+1$. Note that $2$ is not always tight, for example when $k = 0$.

\section{Labeled Query  While  Language}
\label{sec:loop_language}
{The language of {\THESYSTEM} is a standard while language with labels to identify different components and with primitives for queries, and equipped with a  trace-based operational semantics which is the main technical tool we will use to define the program's adaptivity.

\vspace{-0.1cm}
{\small
\[
\begin{array}{llll}
\mbox{Arithmetic Expression}
& \aexpr & ::= &
n ~|~ {x} ~|~ \aexpr \oplus_a \aexpr
~|~ \elog \aexpr  ~|~ \esign \aexpr ~|~ \max(a, a) ~|~ \min(a, a)
\\
\mbox{Boolean Expression} & \bexpr & ::= &
\etrue ~|~ \efalse  ~|~ \neg \bexpr
 ~|~ \bexpr \oplus_b \bexpr
~|~ \aexpr \sim \aexpr
\\
\mbox{Expression} & \expr & ::= & v ~|~ \aexpr \sep \bexpr ~|~ [\expr, \dots, \expr]
\\
\mbox{Value}
& v & ::= & { n \sep \etrue \sep \efalse ~|~ [] ~|~ [v, \dots, v]}
\\
\mbox{Query Expression}
& {\qexpr} & ::=
& { \qval ~|~ \aexpr ~|~ \qexpr \oplus_a \qexpr ~|~ \chi[\aexpr]}
\\
\mbox{Query Value} & \qval & ::=
& {n ~|~ \chi[n] ~|~ \qval \oplus_a  \qval ~|~ n \oplus_a  \chi[n]
    ~|~ \chi[n] \oplus_a  n}\\
\mbox{Label}
& l & \in & \mathbb{N} \cup \{\lin, \lex\} \\
\mbox{Labeled Command}
& {c} & ::= &   [\assign {{x}}{ {\expr}}]^{l} ~|~  [\assign {{x} } {{\query(\qexpr)}}]^{l}
~|~ {\ewhile [ \bexpr ]^{l} \edo {c} }
 \\
 &&&
~|~ {c};{c}
~|~ \eif([\bexpr]{}^l , {c}, {c})
~|~ [\eskip]^l
\end{array}
\]
\vspace{-0.1cm}
}

Expressions include
standard arithmetic (with value $n \in \mathbb{N}\cup \{ \infty \}$) and boolean expression, ($\aexpr$ and $\bexpr$) and extended query expressions $\qexpr$.
A query expression $\qexpr$ can be either a simple arithmetic expression $a$, an expression of the form $\chi[\aexpr]$ where $\chi$ represents a row of the database  and  $\aexpr$ represents an index used to identify a specific attribute of the row $\chi$, a combination of two query expressions, $\qexpr \oplus_{a} \qexpr$, or a normal form $\qval$.
For example, the query expression $\chi[3] + 5$  denotes the computation that
obtains
the value in the $3$rd column of $\chi$ in one row and then adds $5$ to it.

Commands are the typical ones from while languages with an additional command $\assign{x}{\query(\qexpr)}$ for query requests which can be used to interrogate
 the database and compute the linear query corresponding to $\qexpr$.
Each command is annotated with a label $l$, and we will use natural numbers as labels to record
the location of each command, so that we can uniquely identify them.
We also have a set of labels $\ldom$, a set $\mathcal{LV}$  of labeled variables (simply variables with a label), and a set $\cdom$ of all the programs.
We denote by $\mathbb{LV}(c)$ the set of labeled variables assigned in an assignment command in the program $c$.
We denote by  $\qvar(c)$ the set of labeled variables that are assigned the result of a query in the program $c$.
 \highlight{We provide the table of notations in Table~\ref{tb:notation} for quick reference.}

\subsection{ Trace-based Operational Semantics}

We use a trace-based operational semantics tracking the history of program execution. The operational semantics is parameterized by a database that can be accessed only through queries. Since this database is fixed, we omit it from the semantics but it is important to keep in mind that this database exists and it is what allows us to evaluate queries.
A \emph{trace}
$\trace$ is a list of \emph{events} generated when executing specific commands. We denote by $\mathcal{T}$ the set of traces and we will use list notation for traces,
 where $[]$ is the empty trace, the operator $\traceadd$ combines an event and a trace in a new trace,
and the operator $\tracecat$ concatenates two traces.
We have two kinds of events: \emph{assignment events} and \emph{testing events}.
Each event consists of a quadruple,
and we use $\eventset^{\asn}$ and $\eventset^{\test}$ to denote the set of all assignment events and testing events, respectively.
\begin{center}
  $ \begin{array}{lllll}
  \mbox{Event}
  & \event & ::= &
  ({x}, l, v, \bullet) ~|~ ({x}, l, v, \qval)  & \mbox{Assignment Event} \\
  &&& ~|~(\bexpr, l, v, \bullet)  & \mbox{Testing Event}
  \\
  \end{array}$
  \end{center}
An assignment event tracks the execution of an assignment  or a query request and consists of the assigned variable, the label of the command that generates it, the value assigned to the variable, and the normal form  $\qval$ of the query expression that has been requested, if this command is a query request, otherwise a default value $\bullet$.
A testing event tracks the execution of an if or while command and consists of the guard of the command, the label, the result of evaluating the guard, the last element $\bullet$.
 We use the operator $\env (\trace) x$ to fetch the latest value assigned to  $x$ in the trace $\trace$, the operator
$\vcounter$ to count the occurrence of a labeled variable in the trace. \highlight{ The function $\kw{lastVal}(\tau, x)$ mentioned in Section~\ref{sec:adaptivity-informal} can be
expressed as $\lambda \trace. \env (\trace) x$. For any initial trace $\trace$, $\kw{lastVal}(\tau, x)$ returns the latest value of $x$ in $\trace$.
}
We denote by $\tlabel(\trace) \subseteq \ldom$ the set of the labels occurring in $\trace$.
Finally, we use $\mathcal{T}_0(c) \subseteq \mathcal{T}$ to denote the set of \emph{initial traces}, the ones
which assign a value to the input variables. We use $\eventset$ to denote the set of all events.

The trace-based operational semantics is described in terms of a small step evaluation relation
$\config{c, \trace} \to \config{c', \trace}'$  describing how a configuration program-trace evaluates to another
configuration program-state. The rules for the operational semantics are described in Fig.~\ref{fig:os}.
The rules for assignment and query generate assignment events, while the rules for while and if generate testing events.
The rules for the standard while language constructs correspond to the usual rules extended to deal with traces.
We have relations $\config{\trace, \expr} \earrow v $  and $\config{\trace, \bexpr} \barrow v $  to evaluate expressions and boolean expressions, respectively. Their definitions are in the supplementary material.
The only rule that is non-standard is the $\textbf{query}$ rule. When evaluating a query, the query expression $\qexpr$ is first simplified to its normal form $\alpha$ using an evaluation relation $\config{\trace, \qexpr} \qarrow \qval$.
Then normal form $\qval$ characterizes the linear query that is run against the database. The query result $v$ is the expected value of the function $\lambda \chi.\qval$ applied to each row of the dataset. We summarize this process with the notation $\query(\qval) = v$ in the rule $\textbf{query}$.
Once the answer of the query is computed, the rules record all the needed information in the trace.  We will use $\to^*$ for the reflexive and transitive closure of $\to$.
The query expression evaluation relation  $\config{\trace, \qexpr} \qarrow \qval$ is defined by the following rules which reduce a query expression to its normal form.
{\small
\begin{mathpar}
\inferrule{
  \config{\trace, \aexpr} \aarrow n
}{
 \config{\trace,  \aexpr}
 \qarrow n
}
\and
\inferrule{
  \config{\trace, \qexpr_1} \qarrow \qval_1
  \and
  \config{\trace, \qexpr_2} \qarrow \qval_2
}{
 \config{\trace,  \qexpr_1 \oplus_a \qexpr_2}
 \qarrow \qval_1 \oplus_a \qval_2
}
\and
\inferrule{
  \config{\trace, \aexpr} \aarrow n
}{
 \config{\trace, \chi[\aexpr]} \qarrow \chi[n]
}
\and
\inferrule{
  \empty
}{
 \config{\trace,  \qval}
 \qarrow \qval
}
 \end{mathpar}
 }

{\footnotesize %
\begin{figure}
\begin{mathpar}
\boxed{
\mbox{Command $\times$ Trace}
\xrightarrow{}
\mbox{Command $\times$ Trace}
}
\and
\boxed{\config{{c, \trace}}
\xrightarrow{}
\config{{c',  \trace'}}
}
\\
{
\inferrule
{
 \trace, \qexpr \qarrow \qval
 \and
\query(\qval) = v
\and
\event = ({x}, l, v, \qval)
}
{
\config{{[\assign{x}{\query(\qexpr)}]^l, \trace}}
\xrightarrow{}
\config{{\clabel{\eskip}^l,  \trace \traceadd \event} }
}
~\textbf{query}
}
\and
\inferrule
{
 \trace, b \barrow \etrue
 \and
 \event = (b, l, \etrue, \bullet)
}
{
\config{{\ewhile [b]^{l} \edo c, \trace}}
\xrightarrow{}
\config{{
c; \ewhile [b]^{l} \edo c,
\trace \traceadd \event}}
}
~\textbf{while-t}
\end{mathpar}
  \vspace{-0.5cm}
    \caption{Trace-based Operational Semantics for Language.}
    \label{fig:os}
  \vspace{-0.1cm}
\end{figure}
}

    \begin{table}
      \caption{\highlight{Table of Notations}}
      \vspace{-0.5cm}
      \label{tb:notation}
      \begin{center}
        \begin{tabular}{| c |c |c| c| }
          \hline
          $\mathcal{LV}$   & universe of labeled variables  & $\qvar(c)$ & labeled query variables in $c$\\
          $\cdom$  & set of all programs &  $\trace$ &   trace, a list of \emph{events}\\
          $\mathcal{T}$  &  set of traces &  $\trace \traceadd \event$  & combine a trace and an event  \\
          $\ldom$ & set of labels  & $\trace \tracecat \trace'$ &  trace concatenation \\
          $\tlabel(\trace) $  &set of labels occurring in $\trace$  &  $ \env (\trace) x$  & fetch latest value of  $x$ in a given $\trace$ \\
          $\mathcal{T}_0(c) $ &  set of \emph{initial traces} & $\kw{lastVal} (\trace, x)$  & fetch latest value of  $x$ in any $\trace$\\
          $\mathbb{LV}(c)$  & labeled variables in $c$ & $\vcounter(\trace, x^i)$ & occurrence of $x^i$ in the trace $\trace$\\
          \hline
        \end{tabular}
        \end{center}
        \vspace{-0.5cm}
      \end{table}
}
\section{Definition of Adaptivity}
\label{sec:adaptivity}
 In this section, we formally present the definition of adaptivity for a given program. We first define a dependency relation between program variables, then define a semantics-based dependency graph, and finally look at the walk visiting the maximal number of query vertices in this graph.

\subsection{May-dependency between variables}
\label{sec:dep}
We are interested in defining a notion of dependency between program variables since assigned variables are a good proxy to study dependencies between queries---we can recover query requests from variables associated with queries. We consider dependencies that can be generated by either data or control flow.
For example, in the program
\[c_1 =[\assign{x}{\query(\chi[2])}]^1 ;[\assign{y}{\query(\chi[3] + x)}]^2\]
the query $\query(\chi[3] + x)$  depends on the query $\query(\chi[2]))$ through a \emph{value dependency} via  $x^1$.
Conversely, in the program
\[c_2 = [\assign{x}{\query(\chi[1])}]^1 ; \eif( [x > 2]^2 , [\assign{y}{\query(\chi[2])}]^3, [\eskip]^4 )\]
the query $\query(\chi[2])$  depends on the query $\query(\chi[1])$ via the \emph{control dependency} of the guard of the if command involving the labeled variable $x^1$.

To define dependency between program variables we will consider two events that are generated from the same command, hence they have the same variable name or boolean expression and label, but have either different values or different query expressions.

\begin{defn}
\label{def:diff}
\highlight{Two colocated events $\event_1, \event_2 $ differ in their values,  if they are either,
\begin{enumerate}
  \item  of the form $\event_1 = (x, l, v_1, \bullet)$ and $event_1 = (x, l, v_2, \bullet)$, for a common variable (or Boolean expression) $x$ and program location $l$, and $v_1 \neq v_2$.
  \item  or of the form $\event_1 = (x, l, v_1, q_1)$ and $\event_1 = (x, l, v_2, q_2)$, where $q_1 \neq \bullet, q_2 \neq \bullet,$ and $q_1 \neq_q q2$.
\end{enumerate}
denoted as $\diff(\event_1, \event_2)$ as follows:}
{\small
\begin{subequations}
\begin{align}
& \pi_1(\event_1) = \pi_1(\event_2)
  \land
  \pi_2(\event_1) = \pi_2(\event_2) \\
& \land
  \big(
   (\pi_3(\event_1) \neq \pi_3(\event_2)
  \land
  \pi_{4}(\event_1) = \pi_{4}(\event_2) = \bullet )
  \lor
  (\pi_4(\event_1) \neq \bullet
  \land
  \pi_4(\event_2) \neq \bullet
  \land
  \pi_{4}(\event_1) \neq_q \pi_{4}(\event_2))
  \big)
\end{align}
\label{eq:diff}
\end{subequations}
}
where $\qexpr_1 =_{q} \qexpr_2$ denotes the semantics equivalence between query values\footnote{The formal definition is in the supplementary material},
and $\pi_i$ projects the $i$-th element from the quadruple of an event.
\end{defn}

\highlight{In the above definiton, we use $\neq_q$ to show the inequality between two query expressions.}
\highlight{For instance, in the running program in Fig.~\ref{fig:overview-example}(a), the query request command at line $3$, $\clabel{\assign{x}{\query(\chi[j] \cdot \chi[k])} }^{3}$
generates two events $\event_1 = (x, 3, 0, \chi[0] \cdot \chi[1])$ and $\event_2 = (x, 3, 0, \chi[0] \cdot \chi[2])$
given different inputs $k = 1$ and $k = 2$. The two query expressions $\chi[0] \cdot \chi[1]$ and $\chi[0] \cdot \chi[2]$ represent two
varied query requests and can not be regarded as identical, so we have  $\chi[0] \cdot \chi[1] \neq_q \chi[0] \cdot \chi[2]$.
}
Even though the two events have the same query results($\pi_3(\event_1) = \pi_3(\event_2)$), they are still different by our definition.

We can now define when an event \emph{may depend} on another one in Definition~\ref{def:event_dep}\footnote{We consider here dependencies between assignment events. This simplifies the definition and is enough for the stating the following definitions. The full definition is in the supplementary material.}.

There are several components in Definition~\ref{def:event_dep}. The part with label (2a) requires that $\event_1$ and $\event_1'$ differ in their values ($\diff(\event_1, \event_1')$).
The next two parts (2b) and (2c) capture the value dependency and control dependency, respectively.
As in the literature on non-interference, and following~\cite{Cousot19a}, we formulate these dependencies as relational properties, in terms of two different traces of execution.
We force these two traces to differ by using the event $\event_1$ in one and $\event_1'$ in the other.
For the value dependency we check whether the change also creates a change in the value of $\event_2$ or not. We additionally check that the two events we consider appear the same number of times in the two traces - this to make sure that if the events are generated by assignments in a loop, we consider the same iteration.
For the control dependency we check whether the change of the value in the assignment of event $\event_1$ affects the occurrence of event $\event_2$ or not.
For this we require the presence of a test event whose value is affected by the change in $\event_1$
in order to guarantee that the computation goes through a control flow guard.
Similarly to the previous condition, we additionally check that the two test events we consider appear the same number of times in the two traces.

\begin{defn}[Event May-Dependency]
  \label{def:event_dep}
  \vspace{-0.1cm}
  \highlight{An assignment event $\event_2\in \eventset^{\asn}$ \emph{may-depend} on another assignment event $\event_1 \in \eventset^{\asn}$ in a program ${c}$ \highlight{with a witness trace $\trace$}
  , if and only if there exists an assignment event $\event_1'$ which colocate $\event_1$ but differs in values, either
  \begin{enumerate}
    \item replacing $\event_1$ with $\event_1'$ in the execution trace will trigger a new execution trace whose newly generated event $\event_2'$ that colocateds $\event_2$ differs in vaue with $\event_2$, which corresponds to capturing the data dependency.
    \item replacing $\event_1$ with $\event_1'$ in the execution trace will trigger a new execution trace in which a test event after $\event_1$ in the original trace will be flipped , which corresponds to capturing the control dependency.
  \end{enumerate}
  It is denoted
  $\eventdep(\event_1, \event_2, \highlight{\trace},  c)$}
  \begin{subequations}
  {\small
  \vspace{-0.1cm}
  \begin{align}
  &
  \exists \trace, \trace_0, \trace_1, \trace' \in \mathcal{T},\event_1' \in \eventset^{\asn}, {c}_1, {c}_2  \in \cdom  \st \diff(\event_1, \event_1') \land \\
  &
  \quad \exists  \event_2' \in \eventset \st
  \left(
  \begin{array}{ll}
    & \config{{c}, \trace_0} \rightarrow^{*}
    \config{{c}_1, \trace_1 \tracecat [\event_1]}  \rightarrow^{*}
    \config{{c}_2,  \trace_1 \tracecat [\event_1] \tracecat \trace \tracecat [\event_2] }
     \\
     \bigwedge &
     \config{{c}_1, \trace_1 \tracecat [\event_1']}  \rightarrow^{*}
      \config{{c}_2,  \trace_1 \tracecat[ \event_1'] \tracecat \trace' \tracecat [\event_2'] }
    \\
    \bigwedge &
    \diff(\event_2,\event_2' ) \land
    \vcounter(\trace, \pi_2(\event_2))
    =
    \vcounter(\trace', \pi_2(\event_2'))\\
    \end{array}
    \right)\\
    &
    \quad
    \lor
    \left(
    \begin{array}{l}
    \exists \trace_3, \trace_3'  \in \mathcal{T}, \event_b \in \eventset^{\test} \st
    \\
     \quad \config{{c}, \trace_0} \rightarrow^{*} \config{{c}_1, \trace_1 \tracecat [\event_1]}  \rightarrow^{*}
     \config{c_2,  \trace_1 \tracecat [\event_1] \tracecat
     \trace \tracecat [\event_b] \tracecat  \trace_3}
  \\ \quad \land
  \config{{c}_1, \trace_1 \tracecat [\event_1']}  \rightarrow^{*}
  \config{c_2,  \trace_1 \tracecat [\event_1'] \tracecat \trace' \tracecat [(\neg \event_b)] \tracecat \trace_3'}
  \\
  \quad \land \tlabel({\trace_3}) \cap \tlabel({\trace_3'})
  = \emptyset
  \land \vcounter(\trace', \pi_2(\event_b)) = \vcounter(\trace, \pi_2(\event_b))
      \land \event_2 \in \trace_3
      \land \event_2 \not\in \trace_3'
    \end{array}
    \right)
  \end{align}
  }
  \label{eq:eventdep}
  \end{subequations}
  \vspace{-0.3cm}
  \end{defn}

In the running example in Fig.~\ref{fig:overview-example}(b), given the initial trace $[ (k, \lin, 1, \bullet)]$,
the program execution generates an
event $\event_1 = (x, 3, v_1, \chi[0] \cdot \chi[1])$ for
$\clabel{\assign{x}{\query(\chi[j] \cdot \chi[k])} }^{3}$,
and another event
$\event_2 = (l, 6, v_2, \chi[1] \cdot v_1 )$ for
$\clabel{\assign{l}{\query(\chi[k] \cdot a)} }^{6}$.
To check a may-dependency, we replace $\event_1$ with another event $\event_1' = (x, 3, v_1', \chi[0] \cdot \chi[1])$
where $v_1 \neq v_1'$.
We then continue to execute the program. The $6^{th}$ command generates $\event_2 = (l, 3, v_2, \chi[1] \cdot v_1' )$.
Since $\pi_4(\event_2) \neq_1 \pi_4(\event_2')$, we have  $\diff(\event_2, \event_2')$ according to Definition~\ref{def:diff} and then the dependency relation $\eventdep(\event_1, \event_2, \kw{twoRounds(k)})$.

 The dependency relation to variables can be extended by considering all the assignment events generated during the program's execution. Notice two variables can be the same in the following definition, this allows us to capture self-dependencies.
 \begin{defn}[Variable May-Dependency]
  \label{def:var_dep}
  \vspace{-0.2cm}
 \highlight{ A variable ${x}_2$ assigned at location $l_2$ of a program $c$  \emph{may-depend} on the
  variable ${x}_1$ assigned at location $l_1$ in the program ${c}$,
  if and only if there exist two assignment events $\event_1$ and $\event_2$ associated with ${x}_1^{l_1}$ and ${x}_2^{l_2}$ respectively and a witness trace $\trace$,
  and $\event_2$ \emph{may-depend} on $\event_1$ with this witness trace $\trace$ in this program $c$,
   denoted as
  $\vardep({x}_1^{l_1}, {x}_2^{l_2}, {c})$ as follows.
 }
\begin{center}
$
{\small   \begin{array}{l}
\exists \event_1, \event_2 \in \eventset^{\asn}, \trace \in \mathcal{T} \st
\pi_{1}{(\event_1)}^{\pi_{2}{(\event_1)}} = {x}_1^{l_1}
\land
\pi_{1}{(\event_2)}^{\pi_{2}{(\event_2)}} = {x}_2^{l_2}
\land
\eventdep(\event_1, \event_2, \trace, c)
  \end{array}
}%
$
\vspace{-0.2cm}
\end{center}
  \end{defn}

\subsection{Semantics-based Dependency Graph}
\label{sec:design_choice}
\begin{defn}[Semantics-based Dependency Graph]
  \label{def:trace_graph}
  Given a program ${c}$,
  its \emph{semantics-based dependency graph}
  $\traceG({c}) = (\traceV({c}), \traceE({c}), \traceW({c}), \traceF({c}))$ \highlight{consists of four components: vertices are all the
  labelled variables in the program $c$; directed edges are those pairs of vertices that one may depend on the other; weight $w$ of any vertex $x^l$
  is a function which takes an intial trace $\trace_0$ as input and returns the number of times $x^l$ appears in the execution trace of the program $c$ when starting from $\trace_0$; query annotation on any vertex
  can be see as a flag indicating if the command associated with this vertex is a query assignment or not.
  }
  The foraml definiton is defined as follows,
  {\small
  \[
  \begin{array}{lll}
    \text{Vertices} &
    \traceV({c}) & := \left\{
    x^l
    ~ \middle\vert ~ x^l \in \lvar(c)
    \right\}
    \\
    \text{Directed Edges} &
    \traceE({c}) & :=
    \left\{
    (x^i, y^j)
    ~ \middle\vert ~
    x^i, y^j \in \lvar(c) \land \vardep(x^i, y^j, c)
    \right\}
    \\
    \text{Weights} &
    \traceW({c}) & :=
    \{
    (x^l, w)
    ~ \vert ~
    w : \tdom_0(c) \to \mathbb{N}
    \land
    x^l \in \lvar(c) \land
    \forall \trace_0 \in \tdom_0(c), \trace' \in \tdom, l' \st
    \\ & & \qquad \qquad
    \config{{c}, \trace_0} \to^{*}
    \config{\clabel{\eskip}^{l'}, \trace_0 \tracecat \trace'}
    \land w(\trace_0) = \vcounter(\trace', l) \}
     \\
    \text{Query Annotations} &
    \traceF({c}) & :=
  \left\{(x^l, n)
  ~ \middle\vert ~
   x^l \in \lvar(c) \land
  (n = 1 \Leftrightarrow x^l \in \qvar(c))\land ( n = 0 \Leftrightarrow  x^l \notin \qvar(c))
  \right\}
  \end{array},
  \]
  A semantics-based dependency graph $\traceG({c})= (\traceV({c}), \traceE({c}), \traceW({c}), \traceF({c}))$ is \emph{well-formed} if and only if $ \{x^l \ |\ (x^l,w)\in \traceW({c})\} = \traceV({c}) $.
  }
  \end{defn}
We can now define the \emph{semantics-based dependency graph} of a program $c$ in Definition~\ref{def:trace_graph}. We want this graph to combine quantitative reachability information with dependency information.

 Vertices and query annotations are just read out from the program $c$. We have an edge in $\traceE(c)$ if we have a may-dependency between two labeled variables in $c$.
A weight function $w \in \traceW(c)$ is a function that for every starting trace $\trace_0 \in \mathcal{T}_0(c)$
gives the number of times the assignment of the corresponding vertex $x^l$ is visited. Notice that weight functions are total and with range $\mathbb{N}$. This means that if a program $c$ has some non-terminating behavior, the set $\traceW(c)$ will be empty.
To rule out this situation, we consider as well-formed only graphs with a weight for every vertex.
In the rest of the paper we  implicitly consider only well-formed semantics-based dependency graphs.

  \subsection{Adaptivity of a Program}
  \label{sec:sematnic_adaptivity}
 This notion of adaptivity is formulated in terms of an initial trace, specifying the value of the input variables, and of the walk on the graph $\traceG({c})$ which has the largest number of query requests.

In Fig.~\ref{fig:overview-example}(b), $\lambda \trace_0 \st (l^6 \to x^3)$ is a walk with two vertices and where each vertex is visited only once.
With the assumption that $k \geq 1$, $\lambda \trace_0 \st (l^6 \to {a^5 \to a^5 \to \ldots} \to x^3)$ is a walk where the vertex $a^5$ is visited $\env(\trace) k$ times.
However, $\lambda \trace_0 \st (l^6 \to a^5 \to x^3 \to x^3)$ is not a walk because there is no edge from $x^3$ to $x^3$.
The formal defintion of a walk is the following:
\begin{defn}[Walk]
\label{def:finitewalk}
Given a well-formed program $c$ with its semantics-based dependency graph $\traceG({c}) = (\traceV, \traceE, \traceW, \traceF)$, a \emph{walk} $k$ is a function that maps an initial trace $\trace_0$ to a sequence of vertices $(v_1, \ldots, v_{n})$
for which there is a sequence of edges $(e_1 \ldots e_{n - 1})$  satisfying
\begin{itemize}
\item $e_i = (v_{i},v_{i + 1}) \in \traceE$ for every $1 \leq i < n$,
\item and $v_i$ appears in $(v_1, \ldots, v_{n})$ at most $w_i(\trace_0)$ times for every $v_i \in \traceV$ and $(v_i, w_i) \in \traceW$.
\end{itemize}
We denote by $\walks(\traceG(c))$
the set of all the walks $k$ in $\traceG(c)$.
\end{defn}
Because for the adaptivity
we are interested in the dependency between queries,
we calculate a special ``length'' of a walk, the \emph{query length},  by counting only the vertices
corresponding to queries.
\begin{defn}[Query Length]
\label{def:qlen}
Given
the semantics-based dependency graph $\traceG({c})$ of a well-formed program $c$,
 and a \emph{walk}
 $k \in \walks(\traceG(c))$,
the \emph{query length} of $k$ is a function $\qlen(k):\tdom_0(c) \to \mathbb{N}$ that
given an initial trace $\trace_0\in \tdom_0(c)$
gives
the number of vertices that correspond to query variables in the vertex sequence of $k(\trace_0)$.
It is defined as follows.
\begin{center}
   $
  \qlen(k) = \lambda \trace_0 \st |\big( v \mid v \in (v_1, \ldots, v_{n}) \land (v, 1) \in \traceF(c)
  \land k(\trace_0) = (v_1, \ldots, v_{n}) \big)|,
$
\end{center}
where the notation $| (\ldots) |$ gives the number of vertices in a sequence.
\end{defn}

 We can now define the adaptivity of a well-formed program as follows.
\begin{defn}
    [Adaptivity of a Program]
    \label{def:trace_adapt}
    Given a well-formed program ${c}$,
    its adaptivity $A(c)$ is a function
    $A(c) : \tdom_0(c)\to \mathbb{N}$
    defined as follows.
\begin{center}
$
    A(c) = \lambda \trace_0 \st \max \big
    \{ \qlen(k)(\trace_0) \mid k \in \walks(\traceG(c)) \big \}
$
\end{center}
\end{defn}

\section{The Adaptivity Analysis Algorithm - {\THESYSTEM}}
\label{sec:algorithm}
{In this section, we present our program analysis {\THESYSTEM} for
computing an upper bound on the \emph{adaptivity} of a given program
$c$.
 The high-level idea behind {\THESYSTEM} is the following:
  \begin{enumerate}
    \item construct
    an \emph{estimated program dependency graph} \progG(c) of a program $c$;
    \item use a reachability-bound algorithm to estimate the weight of every vertex in the graph;
    \item  use an algorithm, which we call AdaptBD, to estimate an upper bound on the number of queries in the walk over the weighted dependency graph that visits the most query vertices. 
  \end{enumerate}
  The \emph{estimated dependency graph} overapproximates the
  semantics-based dependency graph in two dimensions: it
  overapproximates the dependencies between assigned variables (Section~\ref{sec:alg_edgegen}), and, it
  overapproximates the weights (Section~\ref{sec:alg_weightgen}). Thanks to this,
 we prove the upper bound provided by AdaptBD is a sound upper bound
  over the semantic-based adaptivity in Definition~\ref{def:trace_adapt} in Theorem~\ref{thm:adaptalg_soundness}.

\subsection{Weighted Dependency Graph Construction}
Given a program $c$, the set of vertices $\progV(c)$ and query annotations $\progF(c)$ of the \emph{estimated dependency graph} can be computed by simply
scanning the program $c$. These sets can be computed precisely and correspond to
the same sets in the semantics-based dependency graph.
This means that $\progG(c)$ has the same underlying vertex structure as
the semantics-based graph $\traceG(c)$. The differences will be in the sets of edges and weights.

\subsubsection{Edge Construction}
\label{sec:alg_edgegen}
The set of edges $\progE(c)$ is built using a {\em feasible data-flow analysis},
a data flow analysis over
our query language, combined with a reaching definition analysis.
The reaching definition analysis computes a set $\live(l, c)$ of reachable
variables at line $l$ in the program $c$.

The {\em feasible data-flow analysis} computes for every pair $x^i, y^j \in \lvar(c)$
 whether there is a flow from $x^i$ to $y^j$.
 This analysis is based on a relation $\flowsto(x^i, y^j, c)$ built over the sets $\live(l, c)$ for every location $l$. Its formal definiton is in the supplementary material.
It gives us an overapproximation of the \emph{variable may-dependency} relation for direct dependencies (dependencies that do not go through other variables).
We can now identify edges of the dependency graph by computing a transitive closure (through other variables) of the
 $\flowsto$ relation. There is a directed edge from  $x^i$ to $y^j$ if and only if there is chain of variables
    in the $\flowsto$ relation between $x^i$ and $y^j$, defined as follows:
   \begin{center}
$
\begin{array}{cl}
    \progE(c) \triangleq &
    \{
    (y^j, x^i)  ~ \vert ~ y^j, x^i \in \progV(c)
    \land
      \exists n, z_1^{r_1}, \ldots, z_n^{r_n} \in \lvar(c) \st
    \\
    & \qquad \qquad
      n \geq 0 \land
      \flowsto(x^i,  z_1^{r_1}, c)
      \land \cdots \land \flowsto(z_n^{r_n}, y^j, c)
    \}
    \end{array}
$
\end{center}
We prove that the set $\progE(c)$ soundly approximates the set $\traceG(c)$.
	\begin{lem}[Mapping from Egdes of $\traceG$ to $\progG$]
	\label{lem:edge_map}
	For every program $c$ we have:
   \begin{center}
$
	\begin{array}{l}
	\forall e = (v_1, v_2) \in \traceE(c)
	\st
	\exists e' \in \progE(c) \st e' = (v_1, v_2)
	\end{array}
$
\end{center}
	\end{lem}

  \highlight{
  We provide a proof sketch and the complete proof is in the supplementary material, appendix~C.1.
  \begin{proof}
   This lemma is proved by Theorem~\ref{thm:flowsto_soundness}, saying
    that for arbitrary two labelled variables $x^i, y^j \in \lvar_c$ in the program $c$,
    there exists a static ``$\flowsto$ chain'' from $x^i$ to $y^j$ when
    the variable $x$ at location $l$ may depend on the variable $y$ at location $j$, as $\vardep(x^i, y^j, c)$.
  \end{proof}
\begin{thm}[$\vardep$ implies $\flowsto$]
  \label{thm:flowsto_soundness}
  Given a program ${c}$, for arbitrary two labelled variables  $ x^i, y^j \in \lvar_{{c}}$, if $x$ at location $i$ may depend on
  variable $y$ at location $j$, as $\vardep(x^i, y^j, {c})$,
  then
  there exist a non-empty list of variables [$z_1, \cdots,z_n $]  between locations $i$ and $j$ so that
  there exists a flow from $x^i$ to $y^j$ through these variables.
  \[
  \begin{array}{l}
    \forall x^i, y^j \in \lvar_{{c}}.
    \vardep(x^i, y^j, {c})
    \\ \quad \implies
    \Big( \exists n \in \mathbb{N}, z_1^{r_1}, \ldots, z_n^{r_n} \in \lvar_{{c}} \st n \geq 0 \land
    \flowsto(x^i,  z_1^{r_1}, c)
    \land \cdots \land \flowsto(z_n^{r_n}, y^j, c) \Big)
  \end{array}
  \]
  \end{thm}
  \begin{proof}
    It is sufficient to show two cases:
    \begin{enumerate}
      \item  $x^i$ is directly used in the expression of the assignment command associated to $y^j$, or a boolean
      expression of the guard for a if or while command with the assignment command associated to $y^j$ showing up in the body of that command,
      we call it, $x^i$ \emph{directly} flows to $y^j$,
      i.e.,  $ \flowsto(x^i, y^j, {c})$.
      \item there exists another labelled variable $z^l$ with \emph{variable May-Dependency} relation on $x^i$ and
      $z^l$ directly flows to $y^j$, where the \emph{variable May-Dependency} relation between $x^i$ and $z^l$ implies a ``sub $\flowsto$-chain'' from $x^i$ to $z^l$,
      formally,
      $\Big(
        \exists z^l \in \lvar_{{c}}.
      (\vardep(x^i, z^l, {c})
      \implies
        \exists n \in \mathbb{N}, z_1^{r_1}, \ldots, z_n^{r_n} \in \lvar_{{c}} \st n \geq 0 \land
      \flowsto(x^i,  z_1^{r_1}, c)$
      $\land \cdots \land \flowsto(z_n^{r_n}, z^l, c)
      )
      \land  \flowsto(z^l, y^j, {c})
      \Big)$.
    \end{enumerate}
The first case is proved by the base case of induction on the trace, and the inversion on expression and event.
\\
The second case is proved by proving the following two lemmas:
\begin{lem}[The Multiple-Steps Event Dependency Inversion]
  \label{lem:depevents_exist}
For every hidden database $D$ , a program $c$,a trace $\trace$, and two assignment events
$\event_1, \event_2 \in \eventset^{\asn}$,
if the trace $\trace$ has the form $\trace = [\event_1] \tracecat \trace' \tracecat [\event_2]$ with $\trace' \in \mathcal{T}$,
and the event $\event_2$ may depend on the event $\event_1$,
then there exists a flow from the the labelled variable of $\event_1$ to
the labelled variable of $\event_2$.
 Otherwise there exists
$\event \in \trace'$ such that
$\left(
    \eventdep(\event_1, \event, \trace[\event_1:\event], c, D)
\land
\flowsto(\pi_1(\event)^{\pi_2(\event)}, \pi_1(\event_2)^{\pi_2(\event_2)}, c)
\right)$.
\end{lem}
\begin{thm}[$\eventdep$ implies $\flowsto$]
  \label{thm:flowsto_event_soundness}
  For every hidden database $D$, a program $c$,a trace $\trace$, and two assignment events
$\event_1, \event_2 \in \eventset^{\asn}$,
if the trace $\trace$ has the form $\trace = [\event_1] \tracecat \trace' \tracecat [\event_2]$ with $\trace' \in \mathcal{T}$,
and the event $\event_2$ may depend on the event $\event_1$,  then there exists
    $z_1^{r_1}, \ldots, z_n^{r_n} \in \lvar_{{c}}$ with $n \geq 0$ such that
  $\flowsto(x^i,  z_1^{r_1}, c)
  \land \cdots \land \flowsto(z_n^{r_n}, y^j, c)$
  \end{thm}
\begin{enumerate}
  \item Lemma~\ref{lem:depevents_exist}: existence of a middle event.
  This is proved by showing a contradiction, with detail in supplementary material Appendix~C.2.%
  \item Theorem~\ref{thm:flowsto_event_soundness}: the middle event with a sub-trace implies a ``sub $\flowsto$-chain''.
  This is proved by induction on the trace $\trace$ and applying the induction hypothesis, in appendix C.3
\end{enumerate}
\end{proof}
}
\subsubsection{Weight Estimation}
\label{sec:alg_weightgen}

The weight of every vertex is obtained as an upper bound on
the number of times the corresponding command can be executed,
which can be statically estimated by \emph{reachability-bound analysis}~\cite{GulwaniZ10}.
Our reachability-bound analysis adapts ideas from previous works~\cite{ZulegerGSV11,SinnZV14,sinn2017complexity}
to our query language setting. First, we need to
construct an abstract transition graph $\absG(c)$ of a program $c$, a
graph with the set of labels of program points in $c$  (including a
label $\lex$ for the exit point) as the set of
vertices $\absV(c)$, and with the set of transitions in $c$ as the set of
edges $\absE(c)$. Each edge of the graph is annotated with either
the symbol $\top$, a boolean expression or a \emph{difference
constraint}~\cite{sinn2017complexity}. A difference constraint is an
inequality of the form $x' \leq y + v$ {or $x' \leq v$}
where $x, y$ are variables and $v \in \constdom$ is a symbolic
constant: either a natural number, the symbol $\infty$, an input
variable or a symbol $Q_m$ representing a query request. Difference constraints describe
quantitative conditions that are guaranteed by the execution of the command with
the corresponding label. We denote by
$\dcdom$ the set of difference constraints.

The set $\progW(c)$ of weights for the estimated dependency graph is estimated from
the Abstract Transition Graph $\absG(c)$ of $c$ using reachability-bound analysis.
 Specifically, we use the weight of a vertex with label $l$ as the \emph{symbolic upper bound} on the
 execution times of the command with label $l$ that we obtain using our reachability-bound analysis.
 These symbolic upper bounds are expressions with the input variables as free variables,
 hence they can be used to upper bound the weight functions of the semantics-based dependency graphs.
We denote the reachability-bound of an edge $e = (l, dc, l') \in \absE(c)$ as $\absclr(e, c)$.
More details are in the supplementary material.

\begin{defn}[Weight Estimation]
  \vspace{-0.1cm}
 \label{def:adaptfun-weight}
The estimated weight set $\progW(c)$ of program $c$ is a set of pairs $\progW(c) \in \mathcal{P}(\mathcal{LV} \times \constdom)$.
Each pair maps
a vertex $x^l \in\progV(c)\subseteq  \lvar(c)$ to a symbolic expression over $\constdom$,
as follows.
 $
 \progW(c) \triangleq
 \left\{ (x^l, \hat{w})
\mid
x^l \in \progV(c)
\land
\hat{w} =
\sum\left\{ \absclr(e, c) \middle\vert e \in \absE(c) \land e = (l, \_, \_) \right\}
\right\}.
$
\end{defn}
\begin{thm}[Soundness of the Weight Estimation]
 \label{thm:addweight_soundness}
 \vspace{-0.1cm}
For a program ${c}$ and its estimated weight set $\progW(c)$, for every  $(x^l, w) \in \traceW(c) $,
we have a corresponding estimated weight $(x^l, \hat{w}) \in \progW(c)$ and for every possible
$\vtrace_0 \in \tdom_0(c),
v \in \mathbb{N}$,
if $\config{\vtrace_0, \hat{w}} \earrow v$,
then $v$ is an upper bound on ${w}(\trace_0)$. In symbols:
\[
 \begin{array}{l}
 \forall c \in \cdom, x^l \in \lvar(c),w\in \tdom_0(c)\to \mathbb{N},\trace_0 \in \tdom_0(c),
v \in \mathbb{N}\cup \{\infty\},
\\ \qquad
(x^l, w) \in \traceW(c) \st \Longrightarrow
\exists
(x^l, \hat{w}) \in \progW(c)
\land
\big(
 \config{\vtrace_0, \hat{w}} \earrow v \implies w(\trace_0) \leq v
\big)
\end{array}
\]
\end{thm}
Notice that $\hat{w} \in \constdom$ is an expression over symbols in $\constdom$. In particular, it may contain the input variables and so it may effectively be used as a function of the input - and capture loop bounds in terms of these inputs.
Also in this theorem, the evaluation $\config{\vtrace_0, \hat{w}} \earrow v$ is needed to obtain a concrete value $v$ from the symbolic weight $\hat{w}$ by specifying a value for the input variables via $\vtrace_0$.

\begin{defn}
  [Estimated Dependency Graph]
  \label{def:prog_graph}
  Given a program $c$
  with the feasible data flow relation $\flowsto(x^i, y^j, c)$ for every $x^i, y^j \in \lvar(c)$,
  and reachability-bound, $\absclr(e, c)$ for every edge $e \in \absE(c)$ on the abstract transition graph,
  the estimated dependency graph of $c$ has four components:
  \[\progG(c) = (\progV(c), \progE(c), \progW(c), \progF(c))\]
\end{defn}

\subsection{Adaptivity Upper Bound Computation}
\label{sec:alg_adaptcompute}
We estimate the adaptivity upper bound,
$\progA(c)$ for a program $c$ as the maximal query length over all finite walks
in its \emph{estimated dependency graph}, $\progG({c})$.

\begin{defn}
[{Estimated Adaptivity}]
\label{def:prog_adapt}
{
Given a program ${c}$ and its estimated dependency graph
$\progG({c})$
the estimated adaptivity for $c$ is
\begin{center}
$
\progA({c})
\triangleq \max
\left\{ \qlen(k) \ \mid \  k \in \walks(\progG(c))\right \}.
$
\end{center}
}
\end{defn}

Different from a walk on $\traceG(c)$, a walk $k \in \walks(\progG(c))$ on the graph $\progG(c)$
 does not rely on an initial trace.
 The reason is that we use symbolic expressions over input variables, similarly to what we did for the weights in $\progW(c)$ in the previous section. Thus the adaptivity bound $\progA(c)$ will also be a symbolic arithmetic expression over the input variables. With this symbolic expression, we can prove the upper bound sound with respect to any initial trace.

\begin{thm}[Soundness of $\progA(c)$]
    \label{thm:sound_progadapt}
    For every program $c$,
    its estimated adaptivity is a sound upper bound of its adaptivity.
\begin{center}
$
     \forall \trace_0 \in \tdom_{0}(c), v \in \mathbb{N}\cup \{\infty\} \st
\config{\progA(c), \trace_0} \earrow v \implies A(c)(\trace_0) \leq v.
$
\end{center}
\end{thm}

Symbolic expressions as used in the weight are great to express symbolic bounds but make the direct computation of
a maximal walk harder. Specifically, to directly compute the estimated adaptivity $\progA$,
one has to traverse the estimated dependency graph.
A naive traversing strategy leads to non-termination
because the weight of each vertex in $\progG({c})$
is a symbolic expression containing input variables.
We could try to use a depth first search strategy
using the longest weighted path to approximate
the finite walk with the weight as
its visiting time. However, this approach would  consistently and considerably over-approximate the adaptivity. These challenges
motivated us to develop an algorithm,   $\pathsearch$,  that provides a sound upper bound
on $\progA$.

\subsection{$\pathsearch$ Algorithm}

\subsubsection{$\pathsearch$ overview}
The idea of the $\pathsearch$ algorithm is to reduce the task of computing the
walk with the maximal query length in the dependency graph to the task of
finding the longest weighted path in a simplified dependency graph which happens to be a directed acyclic graph (DAG).
To get the simplified graph, the following steps are performed in $\pathsearch$:
\begin{enumerate}
  \item The original dependency graph is partitioned into
  its strongly connected components (SCCs) using the standard Kosaraju’s algorithm.
  \item  For each SCC, we compute a local adaptivity bound using an helper algorithm $\pathsearch_{\kw{SCC}}$, which we describe in the next section.
  \item  $\pathsearch$ shrinks the estimated dependency graph into a directed acyclic graph
  by reducing each SCC into a vertex with the weight equal to its adaptivity local bound.
\end{enumerate}
After simplifying the graph, $\pathsearch$ uses breadth first search  to find the longest
weighted path on this DAG and return this length as the adaptivity upper bound. Notice that the problem of computing the longest weighted path is in general unfeasible but since we are in a DAG the complexity is O(V+E) where $V$ and $E$ represent the number of vertices and the number of edges
in the simplified graph. The pseudo-code for $\pathsearch$ can be found in the supplementary material. In the next section we present the helper algorithm $\pathsearch_{\kw{SCC}}$.

\subsubsection{Algorithm $\pathsearch_{\kw{SCC}}$}
The high level idea of $\pathsearch_{\kw{SCC}}$ is the following:
to provide an upper bound on the local adaptivity of an Strong Connected Component of a graph,
we need to compute an upper bound on the maximal query length of walks in this SCC.
For the sake of illustrating the algorithm, suppose we have an SCC with just three vertices $a^{3}_{1}$, $b^{2}_{0}$,
$c^{3}_{1}$, where $a$ and $c$ are query vertices with weight $3$ and $b$ is a standard assignment vertex with weight $2$. Suppose the SCC is a cycle with edge $a^{3}_{1} \to b^{2}_{0}$, $b^{2}_{0} \to c^{3}_{1}$ and $c^{3}_{1}\to a^{3}_{1}$.
A walk with maximal query length is for instance: $a \to b \to c \to a \to b \to c \to a$. The query length of this walk is $5$ and it can be approximated by
 first finding a path which starts from a vertex and stops when it reaches the same vertex, such as $a \to b \to c \to a$. Then we can compute a sound upper bound by taking the product of
 the minimal weight along the path ($\min(3,2,3,3)$) and the number of queries of this path, $3$ ($a$ and $c$ are query vertices). In this case, the upper bound is $6$.
Because we are
 traversing an SCC, if we start from one vertex, it will finally go back to the same vertex.
 $\pathsearch_{\kw{SCC}}$ uses the above intuition and traverses an SCC starting from every possible
 vertex, keeps track of the minimal weight (line: 7 in Algorithm~\ref{alg:adaptscc}) and the number of queries (line: 8 in Algorithm~\ref{alg:adaptscc}) and
 stops when it goes back to the same vertex (line: 11 in Algorithm~\ref{alg:adaptscc}). When a cycle is detected,
 the upper bound is updated (line: 12 in Algorithm~\ref{alg:adaptscc}) by taking the maximal between the current upper bound and the product of the minimal weight
 and the number of query vertices along the visited path as we discussed in the example above. When all the traversals starting from
 every vertex of this SCC are finished, the obtained upper bound will be the maximal approximated one.
 The traversal is naturally performed by a Depth First Search, whose complexity is O(V+E), where $V$ and $E$
 represent the number of vertices and edges in this SCC.  $\pathsearch_{\kw{SCC}}$ performs such
 a traversal for every vertex so the total complexity for $\pathsearch_{\kw{SCC}}$ is O((V+E)*V).

{\footnotesize
\begin{algorithm}
            \caption{
            {\small Adaptivity Bound Algorithm on An SCC ({$\kw{\pathsearch_{scc}(c, SCC_i)}$})}
            \label{alg:adaptscc}
            }
            \begin{algorithmic}[1]
              \REQUIRE The program $c$,
              A strong connected component of $\progG(c)$: $ \kw{SCC_i} = (\vertxs_i, \edges_i, \weights_i, \qflag_i)$
            \STATE {\bf init.}
            $\kw{r_{scc}}$: $\mathcal{A}_{\lin}$ List with initial value $0$.
            \STATE {\bf init.}
            $\kw{visited}$ : $\{0, 1\}$ List with initial value $0$;
            $\kw{r}$ : $\mathcal{A}_{\lin}$ List, initial value $\infty$;
            \\ \qquad
            $\kw{flowcapacity}$: $\mathcal{A}_{\lin}$ List, initial value $\infty$;
            $\kw{querynum}$: INT List, initial value $\qflag_i(v)$.
            \STATE {\bf if} $|\vertxs_i| = 1$ and $|\edges_i| = 0$:
            \quad {\bf return}  $\qflag(v)$
            \STATE  {\bf def} {$\kw{dfs(G,s,visited)}$}:
            \STATE \qquad {\bf for} every vertex $v$
            connected by a directed edge from $s$:
            \STATE \qquad \qquad {\bf if} $\kw{visited}[v] = \efalse$:
            \STATE \qquad \qquad \qquad {$\kw{flowcapacity[v] = \min(\weights_i(v), {flowcapacity}[s])}$}; \qquad \qquad \#\{{track minimal weight so far}\}
            \STATE \qquad \qquad \qquad {$\kw{querynum[v] = querynum[s] + \qflag_i(v)}$}; \qquad \qquad \qquad \qquad\#\{{track numbers of queries so far} \}
            \STATE \qquad \qquad \qquad {$\kw{r[v] =  \max(r[v], flowcapacity[v] \times querynum[v]}) $}; \#\{{track adaptivity upper bound so far} \}
            \STATE \qquad \qquad \qquad  $\kw{visited}[v] = \etrue$; 
            \quad $\kw{dfs(G, v, visited)}$;
            \STATE \qquad \qquad {\bf else}: \#\{There is a cycle finished\}
            \STATE \qquad \qquad \qquad
            {\small{$\kw{r[v] =  \max(r[v], r[s] +\min(\weights_i(v), {flowcapacity}[s]) \times (querynum[s] + \qflag_i(v)))}$}};
            \STATE \qquad {\bf return}  $\kw{r[c]}$
            \STATE  {\bf for} every vertex $v$ in $\vertxs_i$:
            \STATE  \qquad initialize the $\kw{visited, r, flowcapacity, querynum}$ with the same value at line:2.
            \STATE  \qquad $\kw{r_{scc} = \max(r_{scc}, dfs(SCC_i, v, \kw{visited} ))}$;
            \RETURN  $\kw{r_{scc}}$
            \end{algorithmic}
            \end{algorithm}
}

The pseudo-code of $\pathsearch_{\kw{SCC}}$ is given as Algorithm~\ref{alg:adaptscc}.
This algorithm takes as input the program $c$ and a $\kw{SCC_i}$ of
$\progG(c)$, and outputs an adaptivity bound for $\kw{SCC_i}$.
If $\kw{SCC_i}$ contains only one vertex, $x^l$ without any edge, $\kw{\pathsearch_{scc}}$ returns the query annotation of $x^l$ as the adaptivity.
If $\kw{SCC_i}$ contains at least one edge,
$\kw{\pathsearch_{scc}}$
recursively computes the adaptivity upper bound on the fly of paths collected through a DFS procedure $\kw{dfs}$ (lines: 4-13). This procedure guarantees that the visiting times of each vertex are upper bounded by its weight, and addresses the approximation challenge,
via two special lists parameters $\kw{flowcapacity}$ and $\kw{querynum}$ (lines:7-11).

$\kw{flowcapacity}$ is a list of symbolic expressions $\mathcal{A}_{in}$ which tracks the minimum weight
when searching a path,
and updates the weight when the path reaches a certain vertex. It guarantees that every vertex on the estimated walk is allowed to be visited at most $\kw{flowcapacity[v]}$ times, and this walk is a valid finite walk.
$\kw{querynum}$ is a list of integer
initialized with the value of the query annotation $\qflag_i(v)$ for every vertex.
It tracks the total number of vertices with query annotation $1$
along the path. It guarantees that $\kw{flowcapacity[v] \times querynum[v]}$ computes an accurate query length
because $\kw{querynum[v]}$ is only the number of the vertices with query annotation $1$.

\highlight{An interesting point of our algorithm is it may produce quadratic adaptivity estimates for examples with
nested while loops, such as $\kw{loop2VD}$ in Table~\ref{tb:adapt-imp} in Section~\ref*{sec:implementation}. In $\kw{loop2VD}$,
it includes a starndard nested while loop of the following form:
$$ \assign{j}{k}; \assign{i}{k}; \ewhile ~[j>0]~ \edo~ \big( \assign{j}{j-1}; ~\ewhile~[i > 0]~\edo~( \assign{i}{i-1}; ~\dots~)~\dots~ \big)  $$
The corresponding dependency graph contains an outside SCC corresponding to the outer loop, and inside this SCC, there
exists an inner SCC corresponding to the inner loop. With the help of the reachability bound analysis, we can have linear bound as
the weight in SCCs, that is $k$ for both loops in the example. According to our Algorithm~\ref{alg:adaptscc}, the inner SCC will be
simplified as one vertex with weight $k$ and the bound for the outer SCC will become $k*k$.       }

{\begin{thm}[Soundness of $\pathsearch$]
    \label{thm:adaptalg_soundness}
    For every program $c$, we have
\begin{center}
$
\pathsearch(\progG({c})) \geq \progA(c).
$
\end{center}
\end{thm}
}

}
\section{Path-sensitivity}
\label{sec:examples}
\label{ex:multipleRounds}
\label{ex:multiRoundsS}

\highlight{In general, $\THESYSTEM$ is not path-sensitive.} We show this in one example, $\kw{multiRoundsOdd(k)}$, presented in Fig.~\ref{fig:multiRoundsOdd}(a). 
The problem witnessed by this example occurs when the control flow is more sophisticated than what the static analysis can handle.
$\kw{multiRoundsOdd}(k)$
has adaptivity $1 + k$ and a  while loop with two paths.
In each iteration, the queries $\clabel{\assign{y}{\query(\chi[x])}}^{5}$
and $\clabel{\assign{p}{\query(\chi[x])}}^{6}$ are based on the results of previous queries stored in $x$.
Only the query answer from $\clabel{\assign{y}{\query(\chi[x])}}^{5}$ in the first branch
is used in the query request command at line $7$, $\clabel{\assign{x}{\query(\chi(\ln(y)))} }^{7}$.
However, this branch is only executed in even iterations ($j = 0, 2, \cdots $).
From the semantics-based dependency graph in Figure~\ref{fig:multiRoundsOdd}(b),
the weight function \highlight{$\kw{lastVal}(\trace, \lceil  \frac{k}{2} \rceil  )$   }for the vertex $y^5$ counts the
number of times $\clabel{\assign{y}{\query(\chi[x])}}^{5}$ is evaluated during the program execution under an initial trace $\trace$, i.e., half of the initial value of $k$ from $\trace$.
However, {\THESYSTEM} fails to realize that all the odd iterations only execute the first branch
and that only even iterations execute the second branch.
So it considers both branches for every iteration when estimating the adaptivity.
In this sense, the weight estimated for $y^5$ and $p^6$ are both
$k$ as in Figure~\ref{fig:multiRoundsOdd}(c).
As a result, {\THESYSTEM} computes $y^5  \to x^7  \to y^5  \to \cdots \to x^7 $
as the walk of maximal query length in Figure~\ref{fig:multiRoundsOdd}(c)
where each vertex is visited $k$ times, and so the estimated adaptivity is $1 + 2 * k$, instead of $1 + k$.

{ \small
\begin{figure}
\centering
    \begin{subfigure}{0.24\textwidth}
\centering
\small{
    \[
    \begin{array}{l}
        \kw{multiRoundsOdd}(k) \triangleq \\
        \clabel{ \assign{j}{k}}^{0} ; \\
        \clabel{ \assign{x}{\query(\chi[0])} }^{1} ; \\
            \ewhile ~ \clabel{j > 0}^{2} ~ \edo ~
            \Big(  \\
             \clabel{\assign{j}{j-1}}^{3} ; \\
             \eif(\clabel{j \% 2 == 0}^{4}, \\
             \clabel{\assign{y}{\query(\chi[x])}}^{5}, \\
             \clabel{\assign{p}{\query(\chi[x])}}^{6});  \\
             \clabel{\assign{x}{\query(\chi(\ln(y)))} }^{7} \Big)
        \end{array}
    \]
}
\vspace{-0.4cm}
\caption{}
    \end{subfigure}
\begin{subfigure}{.36\textwidth}
    \begin{centering}
    \begin{tikzpicture}[scale=\textwidth/15cm,samples=200]
\draw[] (8, 1) circle (0pt) node{{ $x^1: {}^{ \tau \to 1}_{1}$}};
 \draw[] (0, 7) circle (0pt) node{{ $y^5: {}^{ \tau \to \lceil \frac{k}{2}\rceil}_{1} $}};
 \draw[] (0, 4) circle (0pt) node{{ $p^6: {}^{\tau \to  \lfloor \frac{k}{2}\rfloor }_{1} $}};
 \draw[] (0, 1) circle (0pt) node{{ $x^7: {}^{ \tau \to k }_{1}$}};
 \draw[] (8, 7) circle (0pt) node {{$j^0: {}^{\tau \to 1}_{0}$}};
 \draw[] (8, 4) circle (0pt) node {{ $j^3: {}^{ \tau \to k}_{0}$}};
 \draw[  -latex,]  (0, 3.5) -- (0, 1.5) ;
 \draw[  -Straight Barb] (6.5, 4.5) arc (150:-150:1);
 \draw[  -latex] (8, 4.5)  -- (8, 6.5) ;
\draw[ thick,  -latex, densely dotted] (-0.6, 1.5)  to  [out=-220,in=220]  (-0.5, 6.5);
\draw[ thick, -latex, densely dotted]  (0.5, 6.5) to  [out=-30,in=30] (0.6, 1.6) ;
 \draw[ -latex] (2.0, 7)  -- (7, 6.5) ;
 \draw[ -latex] (2.0, 4)  -- (7, 6.5) ;
 \draw[ -latex] (2.0, 1)  -- (7, 6.5) ;
\draw[ -latex] (2.0, 7)  -- (7, 4) ;
\draw[ -latex] (2.0, 4)  -- (7, 4) ;
\draw[ -latex] (2.0, 1)  -- (7, 4) ;
\draw[ -latex] (2.0, 7)  -- (7, 1.5) ;
\draw[ -latex] (2.0, 4)  -- (7, 1.5) ;
\draw[  -latex,] (2.0, 1)  -- (7, 1.5) ;
\draw[ -latex ]  (0, 4.5) -- (0, 6.5) ;
\draw[ -latex ] (0.8, 7.5) arc (220:-100:1);
\draw[ -latex ] (1.2, 1.0) arc (120:-200:1);
\end{tikzpicture}
\vspace{-0.55cm}
\caption{}
    \end{centering}
    \end{subfigure}
    \begin{subfigure}{.36\textwidth}
        \begin{centering}
        \begin{tikzpicture}[scale=\textwidth/15cm,samples=200]
    \draw[] (5, 1) circle (0pt) node{{ $x^1: {}^1_{1}$}};
     \draw[] (0, 7) circle (0pt) node{{ $y^5: {}^{k}_{1}$}};
     \draw[] (0, 4) circle (0pt) node{{ ${p^6: {}^{k}_{1}}$}};
     \draw[] (0, 1) circle (0pt) node{{ ${x^7: {}^{k}_{1}}$}};
     \draw[] (5, 7) circle (0pt) node {{$j^0: {}^{1}_{0}$}};
     \draw[] (5, 4) circle (0pt) node {{ $j^3: {}^{k}_{0}$}};
 \draw[  -latex,]  (0, 3.5) -- (0, 1.5) ;
 \draw[  -Straight Barb] (6.5, 4.5) arc (150:-150:1);
 \draw[  -latex] (5, 4.5)  -- (5, 6.5) ;
 \draw[  -latex,] (1.5, 1)  -- (4, 1) ;
    \draw[thick, -latex, densely dotted] (-0.6, 1.5)  to  [out=-220,in=220]  (-0.5, 6.5);
    \draw[thick, -latex, densely dotted]  (0.5, 6.5) to  [out=-30,in=30] (0.6, 1.6) ;
 \draw[ -latex] (1.5, 7)  -- (4, 6) ;
 \draw[ -latex] (1.5, 4)  -- (4, 6) ;
 \draw[ -latex] (1.5, 1)  -- (4, 6) ;
\draw[ -latex] (1.5, 7)  -- (4, 4) ;
\draw[ -latex] (1.5, 4)  -- (4, 4) ;
\draw[ -latex] (1.5, 1)  -- (4, 4) ;
\draw[ -latex] (1.5, 7)  -- (4, 1) ;
\draw[ -latex] (1.5, 4)  -- (4, 1) ;
\draw[ -latex ]  (0, 4.5) -- (0, 6.5) ;
\draw[ -latex ] (0.8, 7.5) arc (220:-100:1);
\draw[ -latex ] (1.2, 1.0) arc (120:-200:1);
     \end{tikzpicture}
     \vspace{-0.55cm}
     \caption{}
        \end{centering}
        \end{subfigure}
        \vspace{-0.4cm}
\caption{(a) The multiple rounds odd example
(b) The semantics-based dependency graph
(c) The estimated dependency graph from $\THESYSTEM$.}
    \label{fig:multiRoundsOdd}
    \vspace{-0.25cm}
\end{figure}
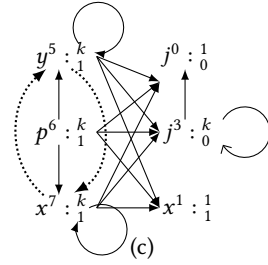
}

\section{Implementation}
\label{sec:implementation}

\paragraph{Set Up}
We implemented $\THESYSTEM$ as a tool that takes a labeled command as input
and outputs two upper bounds on the program adaptivity and the number of query requests, respectively.
This implementation consists of components for the
abstract control flow graph generation,
the edge estimation, and the weight estimation in Ocaml,
and of an implementation of the $\pathsearch$ algorithm in Python.

\paragraph{Accuracy and Performance}
{\scriptsize
\begin {table}[t]
    \caption{Accuracy and Performance Evaluation of {\THESYSTEM} implementation}
    \vspace{-0.4cm}
        \label{tb:adapt-imp}
        \begin{center}
        \centering
{
        \begin{tabular}{| >{\tiny}r | c | c | c | c | c | c | c | c | c | c | c }
         \hline \hline
        \multirow{2}{*}{Program $c$} &
        \multirow{2}{*}{\emph{adaptivity}}
         & \multicolumn{2}{c|}{$\THESYSTEM$}
         & \multirow{2}{*}{L.O.C}
         & \multicolumn{3}{c|}{running time (second)} \\
         \cline{3-4} \cline{6-8}
         & & {$\progA$ } & {$\query$\# } &  & graph & weight & $\pathsearch$ \\
         \cline{6-8}
         \hline \hline
         $  \kw{twoRounds(k)}$ & $2$ &  $2$ & $k+1 $  & 8 & 0.0005 & 0.0017 & 0.0003 \\
         $  \kw{mR(k)}$ & $k$ &  $k$ & $k$  &  10 & 0.0012 & 0.0017  & 0.0002 \\
         $  \kw{lRGD(k, r)}$ & $k$ & $k  $ & $ 2k $  &  10 & 0.0015 & 0.0072  & 0.0002  \\
         $  \kw{mROdd(k)}$ & $1 + k$ &  $2+\max(1,2k)  $ & $1 + 3 k  $  &  10 & 0.0015 & 0.0061  & 0.0002 \\
         $  \kw{mRSingle(k)}$    & $1 + k$ &  $1+ \max(1, k) $ & $1 + k $  &  9 & 0.0011 & 0.0075  & 0.0002 \\
         $  \kw{seqRV()}$ & $4$ & $4$ &  $4$ & 4 & 0.0011 & 0.0003 & 0.0001 \\
         $  \kw{ifVD()}$ & $2$ & $2$ &  $3$ & 5 & 0.0010 & 0.0005  & 0.0001 \\
         $  \kw{ifCD()}$ & $3$ & $3 $ &   $4$  & 5 & 0.0005 & 0.0003   & 0.0001 \\
         $  \kw{loop(k)}$ & $1+k/2$ &   $1 +\max(1, k/2)  $  &  $1+k/2  $ & 7 & 0.0021 & 0.0015 &  0.0001 \\
         $  \kw{loopRV(k)}$ & $1 + 2k$ &  $1 + 2k$ & $2 + 3 k$  &  9 & 0.0016 & 0.0056 & 0.0001  \\
         $  {\kw{loopVCD(k)}} $ & ${1 + 2Q_m}$ &  ${Q_m+\max(1,2Q_m)}$  & $2+2Q_m$   &  6 & 0.0016 & 0.0007 & 0.0001 \\
         $ {\kw{loopMPVCD(k)}}$ & $2+Q_m$ &  $2 + Q_m$  & $2+2Q_m$   &   9 & 0.0017 & 0.0043  & 0.0001 \\
         $  \kw{loop2VD(k)}$ & $2 + k^2$ &   $3 + k^2$ & $1 + k + k^2 $   &  10 & 0.0018 & 0.0126  & 0.0001  \\
         $  \kw{loop2RV(k)}$ & $1 + k +  k^2$ &
         $ 2 + k +  k^2 $
         &  $2 + k + k^2$   &  10 & 0.0017 & 0.0186  & 0.0001  \\
         $  \kw{loop2MV(k)}$ & $1 + 2k $ & $1 + \max(1,2k) $ &  $1 + k + k^2 $  & 10 & 0.0016 & 0.0071  & 0.0001 \\
         $ \kw{loop2MPRV(k)}$ & $1 + k + k^2$ &  $3 + k + k^2  $ &  $2 + 2k + k^2  $  &  10 & 0.019 & 0.0999  & 0.0002 \\
         {$ \kw{loopM(k)}$} & $1 + k$ &  $ 2 + \max(1,2k) $ & $1 + 3k  $  &  9 & 0.0017 & 0.0062  & 0.0001  \\
         {$ \kw{loopM2(k)}$} & $1 + k$ &  $ 2 + k $ & $1 + 3k  $  &  9 & 0.0017 & 0.0062  & 0.0001  \\
         {$\kw{loop2R(k)}$} & $1 + 3k$ &  $1 + 3k $ &  $1 + 3k  + k^2$  &  11 & 0.019 & 0.2669  & 0.0007 \\
         $  \kw{mR(k, N)}$ & $k$ & $ k   $ & $k $   &  27 & 0.0026 & 85.9017  & 0.0004 \\
         $  \kw{mRCom(k)}$ & $2k$ & $  2k $ & $ 2k $   &  46 & 0.0036 & 5104  &  0.0013\\
         $  \kw{seqCom(k)}$ & $12$ & $12  $  & $326 $  &  502 & 0.0426  & 1.2743  & 0.0223 \\
         $  \kw{tRCom(k)}$ & $2$ &  $ * $ &  $ 1 + 5k + 2 k^2$  &  42 & 0.0026 & $*$  & $*$\\
         $  \kw{{jumbo(k)}}$ & $ \max(20, 8+k^2)$ &  $ * $   &   $ {44+k+k^2} $  &  71 & 0.0035 & $*$ &  $*$ \\
         $  {\kw{big(k)}} $ & $22+k+k^2$ &  $* $ &  $121+11k+4k^2 $  &  214 & 0.0175 & $*$ & $*$ \\
         \hline \hline
        \end{tabular}
}
\end{center}
\vspace{-0.2cm}
\end{table}
}

We evaluate our implementation in terms of its accuracy and performance in estimating the adaptivity
through $25$ example programs (summary in Table~\ref{tb:adapt-imp}).
For each example we compare the true accuracy with the estimated adaptivity, and we report the estimated number of queries, the L.O.C and the running time for the graph construction (graph), the weight estimation (weight) and the adaptivity computation ($\pathsearch$).
The set of examples includes the ones we have described before, such as two rounds $\kw{twoRounds(k)}$ in
Fig.~\ref{fig:overview-example}(a) and multiple rounds $ \kw{mR(k)}$ and its variants $\kw{mROdd(k)}$ $\kw{mROdd(k)}$,
a logistic regression algorithm with gradient descent (short for $\kw{LRGD}$),
handcrafted programs from the $6^{th}$ ($\kw{seqRV}$) to $19^{th}$ ($\kw{loop2R}$) based on benchmarks from the work by Gulwani et al.~\cite{GulwaniJK09}. These examples have small sizes but complex structures to test the programs under different situations including
data dependency, control dependency, and the multiple paths nested loop with related counters, etc. The names of these programs obey the convention that,
$\kw{if}$ means there is an if control in the program;
$\kw{loop}$ means there is a while loop and $\kw{loop2}$ represents two levels nested loop in the program;
$\kw{C}$ denotes Control;
$\kw{D}$ for Dependency; $\kw{V}$ for Variable;
$\kw{M}$ for Multiple; $\kw{P}$ for Path and $\kw{R}$ for Related.
The last six programs are synthesized programs composed of the previous programs in order to test the performance limitation when the input program is large.

In terms of accuracy, our tool provides a tight
upper bound on the adaptivity for the most of those examples.
 For the $4^{th}$ program $\kw{mROdd(k)}$, $\THESYSTEM$ outputs an over-approximated upper bound
 $2 + \max(1, 2k)$ because of the path-insensitive nature of the weight estimation algorithm, as discussed in Section~\ref{sec:examples}.
 Similarly for the $17^{th}$ example $\kw{loopM}$.
$\THESYSTEM$ outputs $1 + \max(1, k) $ for the $5^{th}$ program $\kw{mRSingle(k)}$,  tight by our adaptvitiy definition (Definition~\ref{def:trace_adapt}).
However, Definition~\ref{def:trace_adapt} does not capture precisely the accuracy  of $\kw{mRSingle(k)}$,
and the result is still an over-approximation.

In terms of performance, our tool can quickly construct the graph, estimate the weight,
and provide the adaptivity upper bound for the first $19$ examples. \highlight{ $\THESYSTEM$ uses only one second
for the example $\kw{seqCom}$ with the largest size of $502$ lines
of code. However, when it comes to complicated examples with multiple high-level nest loops such as $tRCom$ (two $2-$level and three single-level while loops), $jumbo$ (one $2-$level and two single-level while loops) and $big$ (two $3-$level, three $2-$level and one single-level while loops),
$\THESYSTEM$ will take too long to handle them. As shown in Table~\ref{tb:adapt-imp}, these three examples have the running time $*$ meaning too long to compute.
The performance bottleneck is the reachability bound analysis algorithm used in $\THESYSTEM$.}

{\paragraph{Alternative Implementations}
{\scriptsize
\begin {table}[t]
\vspace{-0.1cm}
    \caption{Accuracy Evaluation of {\THESYSTEM}  Alternative Implementations}
    \vspace{-0.3cm}
        \label{tb:adapt-imp-alternatives}
        \begin{center}
        \centering
{\scriptsize
        \begin{tabular}{ | >{\tiny}l | c | c | c | c | c | c | c | c | c  |c}
        \hline \hline
        \multirow{2}{*}{Program $c$}
        &\multirow{2}{*}{\emph{adaptivity}}
         & \multicolumn{2}{c|}{$\THESYSTEM$-I}
         & \multicolumn{2}{c|}{$\THESYSTEM$-II}
         & {running time} \\
         \cline{3-9}
        & & {$\progA$ } & {$\query$\# } & {$\progA$ } & {$\query$\# }  & $\THESYSTEM$-I \\
         \cline{8-9}
         \hline \hline
          $  \kw{ifCD()}$  &  $3$  & $3$ &   $4$  & \textcolor{red}{$2$} & $4$ & 0.0007 \\
          $ {\kw{loopMPVCD(k)}}$ &  ${2 + Q_m}$  &  $2 + Q_m$  & $2+2Q_m$  & \textcolor{red}{$2$} & $2+2Q_m$ & 0.0020 \\
          {$\kw{loop2R(k)}$} &  $1+3k$ &  \textcolor{red}{$2 + 3k + k^2$} &  \textcolor{red}{$1 + k + k^2$}  &  $2 + 3k + k^2$ &  $1 + k + k^2$ & 0.0199 \\ 
         $  \kw{tRCom(k)}$ &  $2$ &  $ 2$ & $ 1 + 5k + 2 k^2 $  &  $ * $   &   $* $  & 0.0034  \\ 
         $  \kw{{jumbo(k)}}$&  $ \max(20, 8+k^2)$  & $  \max(20, 6+k+k^2)$   &   $ {44+k+k^2} $  &  $ * $   &  $* $ & 0.0123 \\
         $  {\kw{big(k)}} $&  $22+k+k^2$  &   $28 + k + k^2$ &  $121+11k+4k^2 $  &  $ * $   &  $* $  & 0.0181 \\ 
        \hline \hline
        \end{tabular}
}
\end{center}
\vspace{-0.5cm}
\end{table}
}

\highlight{To overcome the performance bottleneck of the underlying implementation of reachability bound analysis,
we implemented {\THESYSTEM}-I, which replaces the reachability bound analysis algorithm
with a simpler one that handles nested while loop effectively but coarsely.
{\THESYSTEM}-I provides a less precise but efficient analysis for these three examples as shown in Table~\ref{tb:adapt-imp-alternatives}.
As for the examples without complicated nest loops, in general, {\THESYSTEM}-I gives tight upper bounds in most examples as {\THESYSTEM} does,
except for the example $\kw{loop2R}$ ($2+3k+k^2$ versus $1+3k$).}

\highlight{To conclude, we treat the alternative implementation {\THESYSTEM}-I as a complement to
$\THESYSTEM$ when the latter cannot handle "complicated" examples with multi-level nested loops.
We do not expect these complex multi-level nested loops to appear in many realistic examples,
but it is an option we can resort to if needed.}
We also have another implementation {\THESYSTEM}-II obtained by removing the control flow analysis. The results are unsound (not upper bounds) for examples such as $\kw{ifCD}$, $ \kw{loopMPVCD(k)} $.
We show the results of these algorithms in Table~\ref{tb:adapt-imp-alternatives}, we highlight in red the results that are imprecise or unsound.

\paragraph{Effectiveness Evaluation}

{\footnotesize
\begin {table}[t]
        \caption{Evaluation of Data Analyses Generalization Error Using {\THESYSTEM}}
    \vspace{-0.4cm}
        \label{tb:adapt-generalization}
        \begin{center}
        \centering
{
        \begin{tabular}{|| >{\tiny}l || c | c || c | l | c | r || c | l | c | r || }
                \hhline{t|:=========== :t:|}
        \multirow{2}{*}{Program $c$}
         & \multicolumn{2}{c||}{$\THESYSTEM$}
         & \multicolumn{4}{c||}{rmse with mechanism(k,m,n=10)}  & \multicolumn{4}{c||}{rmse with mechanism(k,m,n=1000)}  \\
         \hhline{||~--||----||----||}
         & {$\progA$ } & {$\query$\# }  & None  & DS & GS & TS & None & DS & GS & TS \\
         \hline \hline
        $  \kw{ twoRounds }$ & $ 2 $ & $  k + 1 $  & $0.0006$   & {{$0.0015$}} & \textcolor{red}{$0.0004$} & \textbf{0.001}& $0.050$   & \textcolor{red}{\textbf{0.028}} & {$0.031$} & $0.040$  \\
        \hhline{||-||---||-||--||----||}
         $  \kw{ mR}$ & $k$ & $k$  & $0.16$   & $0.1545$  & $0.1087 $ & \textcolor{red}{\textbf{0.1035}}  & $0.066$   & $0.050$ & \textcolor{red}{\textbf{0.036}} & $0.064$  \\
         \hhline{||-||---||-||--||----||}
         $  \kw{ mROdd }$ & $ k $   & $  2 \times k $ & $0.9375$   & $0.999$ & $0.7427$ & \textcolor{red}{\textbf{0.4016}} & $0.211$   & $0.220$ & \textcolor{red}{\textbf{0.059}} & $0.171$  \\
         \hhline{||-||---||-||--||----||}
         $  \kw{ mRSingle }$ & $ k $ & $  k $  & $0.8074$   & \textcolor{red}{$0.3600$} & $ 0.7506$ & {\textbf{0.4036}} & $ 0.761$   & $ 0.758$ & \textcolor{red}{\textbf{0.509}} & $ 0.593$  \\
         \hhline{||-||---||-||--||----||}
         $  \kw{ lRGD }$ & $ k $ & $  2\times k $  & $0.12$   & $0.116$ & $ 0.10 $ & \textcolor{red}{\textbf{0.06}} & $0.216$   & $0.209$ & \textcolor{red}{\textbf{0.014}} & $0.210$  \\
         \hhline{||-||---||-||--||----||}
         $  \kw{ 3DimLRGD }$ & $ k $ & $  3\times k $  & $0.1159$   & $0.10$ & $0.1079$ & \textcolor{red}{\textbf{0.092}} & $0.1966$   & $0.1901$ & \textcolor{red}{\textbf{0.1751}} & $0.1810$  \\
         \hhline{||-||---||-||--||----||}
         $  \kw{ 4DimLRGD }$ & $ k $ & $  4\times k $   & $0.61$   & \textcolor{red}{$0.1080$} & $0.30 $ & \textbf{0.1399} & $0.1112$   & $0.1032$ & \textcolor{red}{\textbf{0.0961}} & $0.1000$  \\
         \hhline{||-||---||-||--||----||}
         $  \kw{ 3DeLRGD }$ & $ k $ & $  3\times k $ & $0.10$  & $0.10$ & $0.06$ & \textcolor{red}{\textbf{0.04}} & $0.1096$   & $0.1056$ & \textcolor{red}{\textbf{0.0098}} & $0.1004$  \\
         \hhline{||-||---||-||--||----||}
         $  \kw{ 4DeLRGD }$ & $ k $ & $  4\times k $  & $0.076$  & $0.0984$ & $0.0719$ & \textcolor{red}{\textbf{0.064}}  & $0.1084$   & $0.1058$ & \textcolor{red}{\textbf{0.1052}} & $0.1055$   \\
         \hhline{||-||---||-||--||----||}
        $\kw{DT}$ & $k$ &  $k$ & $0.0948$  & $0.0948$ & $0.0447$ & \textcolor{red}{\textbf{0.0383}} & $ 1.465$  &  \textcolor{red}{$ 1.283$} & \textbf{1.379 } & {$1.414$}   \\
         \hhline{||-||---||-||--||----||}
         $\kw{LR}$ & $k$ &  $k$ & $0.1316$  & $0.1219$ & $ 0.0893$ & \textcolor{red}{\textbf{0.0632}} & $ 0.152$  &  $ 0.001$ & \textcolor{red}{\textbf{0.001}} & {$0.002$}   \\
         \hhline{||-||---||-||--||----||}
        $\kw{DTOVR}$ & $k \times m$ &  $ k \times m $  & $0.14142$  & $0.1241$ & \textcolor{red}{\textbf{0.0864}} & $ 0.1167$ &  $0.055$ & $0.053$  &  \textcolor{red}{\textbf{0.007}} & $0.036$  \\
        \hhline{||-||---||-||--||----||}
        $\kw{LROVR}$  & $k \times m$ &  $ k \times m $  & $0.0957$  & $0.0917$ & \textcolor{red}{\textbf{0.0815}} & $0.1092$  &  $ 1.000  $  &  $ 1.000 $ & \textcolor{red}{\textbf{ 0.999}} & $ 1.002 $  \\
        \hhline{||-||---||-||--||----||}
        $\kw{RQ}$~\cite{Jamieson2015TheAO} & $ m \times 2^m $ & $  m \times 2^m $  & $0.75$  & \textcolor{red}{$0.89$} & \textbf{6.01} & $1.89$ & $239.0$   & $21.5$ & \textcolor{red}{\textbf{18.557}} & $141.974$   \\
        \hhline{||-||---||-||--||----||}
        $\kw{nDPair}$~\cite{Jamieson2015TheAO} & $ m $ & $  m \times n  $   & $0.34$  & \textcolor{red}{$0.10$} & $0.19$ & \textbf{0.27} & $0.0999$   & $0.0999$ & \textcolor{red}{$0.0970$} & \textbf{0.0999}   \\
        \hhline{||-||---||-||--||----||}
        $\kw{bestArm}$~\cite{Jamieson2015TheAO} & $ n $ & $  m \times n $  & $1.82$  & $1.46$ & $2.67$ & \textcolor{red}{\textbf{0.81}}& $ 2.0452$   & $ 1.3955$ & {{$3.4147$}} & \textcolor{red}{\textbf{1.2871}} \\
        \hhline{||-||---||-||--||----||}
        $\kw{lilUCB}$~\cite{Jamieson2015TheAO} & $ n $ & $ m \times n $ & $3.07$  & $3.11$ & $3.19$ & \textcolor{red}{\textbf{1.79}}& $3.0174$   & $ 3.137$ & {$3.5245$} & \textcolor{red}{\textbf{2.3865}}   \\
        \hhline{|:t=========== t:|}
\end{tabular}
}
\end{center}
\vspace{-0.5cm}
\end{table}
}

In our last evaluation, we consider the effectiveness of {\THESYSTEM} in terms of reducing the generalization error of real-world data analysis programs.
We do this by considering a set of benchmarks including
nine classical data analyses taken from examples in Table~\ref{tb:adapt-imp} instantiated with various parameters,
four programs implementing the algorithms from~\cite{Jamieson2015TheAO}
and four data analysis programs
from \hyperlink{https://github.com/scikit-learn/scikit-learn/tree/main/examples}{sklearn}~\cite{SklearnBenchmark} benchmark.
{We use acronyms for these programs to describe the analyses they perform and fit the evaluation summary in a single table. We use LR for logistic regression, GD for gradient descent, mR for multiple round, DT for decisionTree, RQ for repeatedQuery, OVR for one-vs-rest, De for degree, Dim for dimension, Pair for pairwise. Similarly, we shorten the mechanisms: GS for Gaussian, DS for Data Split,
TS for ThresholdOut.  }

For each program, we show in Table~\ref{tb:adapt-generalization} the
generalization errors when running without a mechanism and running
with different mechanisms over uniformly generated training data.  The
average generalization error of each program is measured by
root-mean-square error (rmse).  The root-mean-square error without any
mechanism is shown in column ``None''.
We mark the best rmse in red.
We mark in bold the rmse
produced by the mechanism chosen by {\THESYSTEM} using the following
heuristics.  According to Theorem~\ref{thm:nonadapt-adapt},\ref{thm:gaussiannoise},
\ref{thm:gaussiannoise2},
we compare the value of $\sqrt{\query \#}$ and
$\progA\sqrt{\log(\query \#)}$ and choose DataSplit if
$\sqrt{\query \#} \gg \progA\sqrt{\log(\query \#)}$ and Gaussian
mechanism if $\sqrt{\query \#} \ll \progA\sqrt{\log(\query \#)}$ and
Thresholdout mechanism if these two quantities are close.
 It is worth
stressing that these are rough heuristics since we are omitting the
constants hidden in these theorems. We leave the problem to identify
the best heuristic to future work.

    According to our heuristics, our examples fall into three categories:
    $\progA$ is much larger than $\query \#$, $\progA$ is close to  $\query \#$,
    and $\progA$ is much smaller than $\query \#$.
    Since both $\progA$ and $\query \#$ are symbolic, the decision of our heuristics relies on
    the instantiation of input parameters such as $k,m, n$. To better evaluate our heuristics, we present two different instantiations $k,m,n=10$ and $k,m, n =1000$ so that
    for the same example, our tool will choose different mechanisms.

We can see in Table~\ref{tb:adapt-generalization} that the mechanism
chosen by our heuristics (marked in bold) is the best in most cases.
In the setting of small parameters $k,m,n=10$, our heuristics
choose most of the times ThresholdOut. Exceptions are examples such as $\kw{DTOVR}$,
$\kw{LROVR}$
from \hyperlink{https://github.com/scikit-learn/scikit-learn/tree/main/examples}{sklearn}~\cite{SklearnBenchmark}
benchmark, and $\kw{RQ}$ from~\cite{Jamieson2015TheAO}, for which our tool
chooses the Gaussian mechanism.  In the setting of large parameters
$k,m,n=1000$,
for most examples our tool chooses the Gaussian mechanism.
For $\kw{twoRounds}$, our heuristics suggest the DataSplit mechanism. For examples $\kw{lilUCB}$,
$\kw{nDPair}$ and $\kw{bestArm}$ from~\cite{Jamieson2015TheAO} within
the \emph{Guess and Check}~\cite{RogersRSSTW20} framework, our tool chooses the Thresholdout mechanism.

\section{Related Work}

\paragraph{Dependency Definitions and Analysis}
There is a vast literature on dependency definitions and dependency analysis.
We consider a semantics definition of dependencies which consider (intraprocedural) data and control dependency~\cite{bilardi1996framework,cytron1991efficiently,pollock1989incremental}.
Our definition is inspired by classical works on traditional dependency analysis~\cite{DenningD77} and noninterference~\cite{GoguenM82a}.
Formally, our definition is similar to the one by \citet{Cousot19a}, which also identifies dependencies by considering differences in two execution traces.
However, Cousot excludes some forms of implicit dependencies, e.g. the ones generated by empty observations,  which instead we consider.
Common tools to study dependencies are dependency graphs~\cite{ferrante1987program}. We use here a semantics-based approach to generating the dependency graph similar, for example, to works by \citet{austin1992dynamic}, \citet{hammer2006dynamic} and \cite{mastroeni2008data}.
Our approach shares some similarities with the use of dependency graphs in works analyzing dependencies between events, e.g. in event programming. \citet{memon2007event} uses an event-flow graph, representing all the possible event interactions, where vertices are GUI event edges represent pairs of events that can be performed immediately one after the other. In a similar way, we use edges to track the may-dependence between variables looking at all the possible interactions.
\citet{arlt2012lightweight} use a weighted edge indicating a dependency between two events, e.g. one event possibly reads data written by the other event, with the weight showing the intensity of the dependency (the quantity of data involved). We also use weights but on vertices and with different meaning, they are functions describing the number of times the vertices can be visited given an initial state.
Differently from all these previous works, we use a dependency graph with quantitative information needed to identify the length of chain of dependencies. Our weight estimation is inspired by  works in complexity analysis and WCET.
Specifically, it is inspired by works on  reachability-bound analysis using program abstraction and invariant inference~\cite{GulwaniZ10, SinnZV17,GulwaniJK09} and work on invariant inference through cost equations and ranking functions~\cite{BrockschmidtEFFG16,AlbertAGP08,AliasDFG10,Flores-MontoyaH14}.
\paragraph{Generalization in Adaptive Data Analysis}
Starting from the works by \citet{DworkFHPRR15} and \citet{HardtU14}, several works have designed methods that ensure generalization for adaptive data analyses~\cite{dwork2015reusable,dwork2015generalization,BassilyNSSSU16,UllmanSNSS18,FeldmanS17,jung2019new,SteinkeZ20,RogersRSSTW20}.
Several of these works drew inspiration from differential privacy, a notion of formal data privacy. By limiting the influence that an individual can have on the result of a data analysis, even in adaptive settings, differential privacy can also be used to limit the influence that a specific data sample can have on the statistical validity of a data analysis. This connection is actually in two directions, as discussed for example by \citet{YeomGFJ18}.
Considering this connection between generalization and privacy, it is not surprising that some of the works on programming language techniques for privacy-preserving data analysis are related to our work.
Adaptive Fuzz~\cite{Winograd-CortHR17} is a programming framework for differential privacy that is designed around the concept of adaptivity.
This framework is based on a typed functional language that distinguishes between several forms of adaptive and non-adaptive composition theorem with the goal of achieving better upper bounds on the privacy cost. Adaptive Fuzz uses a type system and some partial evaluation to guarantee that the programs respect differential privacy. However, it does not include any technique to bound the number of rounds of adaptivity.
\citet{lobo2021programming} propose a language for differential privacy where one can reason about the accuracy of programs in terms of confidence intervals on the error that the use of differential privacy can generate. These are akin to bounds on the generalization error. This language is based on a static analysis which however cannot handle adaptivity.
The way we formalize the access to the data mediated by a mechanism is a reminiscence of how the interaction with an oracle is modeled in the verification of security properties. As an example, the recent works by \citet{BarbosaBGKS21} and \citet{AguirreBGGKS21} use different techniques to track the number of accesses to an oracle. However, reasoning about the number of accesses is easier than estimating the adaptivity of these calls, as we do instead here.
Another important aspect of data analysis is algorithmic fairness, also related to generalization. Several works explore program analysis techniques for fairness for decision-making programs. Albarghouthi et al.~\cite{albarghouthi2017fairsquare} propose automatic program analysis techniques for fairness verification. Albarghouthi and Vinitsky~\cite{albarghouthi2019fairness} propose a runtime monitor tool for fairness and its failure.

\section{Conclusion and future works}
We presented {\THESYSTEM}, a program analysis useful to provide an upper bound on the adaptivity of a data analysis, as well as on the total number of queries asked. This estimation can help data analysts to control the generalization errors of their analyses by choosing different algorithmic techniques based on the adaptivity. Besides, a key contribution of our work is the formalization of the notion of adaptivity for adaptive data analysis. We showed the applicability of our approach by implementing and experimentally evaluating our program analysis.

As future work, we plan to investigate the potential integration of  {\THESYSTEM} in an adaptive data analysis framework like Guess and check by Rogers at al.~\cite{RogersRSSTW20}. As we discussed, this framework is  designed to support adaptive data analyses with limited generalization error. As our experiments show, this framework could benefit from the information provided by {\THESYSTEM} to provide more precise estimates and improved confidence intervals. Another direction is to make
the upper bounds by {\THESYSTEM} more precise by integrating our algorithm with a path-sensitive approach.

\section*{Data Availability}
Data are available from the authors~\cite{artifact}.

\bibliographystyle{ACM-Reference-Format}
\bibliography{main.bib}

\end{document}